\documentclass{elsarticle}

\usepackage[utf8]{inputenc}
\usepackage{natbib}
\usepackage{fancyhdr}
\usepackage{graphicx}
\graphicspath{ {./images/}}

\usepackage{geometry}
\usepackage{xcolor}
\usepackage{indentfirst}
\usepackage{amsthm}
\usepackage{amsmath}
\usepackage{todonotes}
\usepackage{subcaption}
\usepackage[T1]{fontenc}
\usepackage{tikz-cd}
\usepackage{amssymb}

\newtheorem{theorem}{Theorem}[section]
\newtheorem{definition}[theorem]{Definition}
\newtheorem{corollary}{Corollary}[theorem]
\newtheorem{lemma}[theorem]{Lemma}
\newtheorem{example}{Example}[section]
\newtheorem{proposition}[theorem]{Proposition}
\newtheorem*{remark}{Remark}
\newtheorem{problem}[theorem]{Problem}

\title{Equivariant Connections and their applications to Yang-Mills equations}

\author[1]{Driss MAÎTREJEAN \footnote{Ecole Normale Supérieure, 45 rue d'Ulm 75005 Paris France, driss.maitrejean@ens.psl.eu }}
\affiliation[1]{Work done during an internship at University of Vienna, Oskar Morgenstern-Platz 1 1010 Wien Austria and at the Erwin Schrodinger Institute Boltzmanngasse 9 1090 Wien Austria, under the supervision of Roland Donninger, roland.donninger@univie.ac.at}
\date{June 2024}

\begin{document}
\begin{abstract}
    We reduce Yang-Mills equations for $SO^+(p,q)$, $Spin^+(p,q)$ and $SU(n)$ bundles, with constant and isotropic metrics, by developing the concept of $SO^+(p,q)$-equivariance. This allows us to model the electroweak interaction and $SO^+(p,q)$ bundles with a non-linear second order differential equation as well as the weak and strong interaction with a non-linear wave equation.
\end{abstract}
\maketitle

\newpage
\tableofcontents

\newpage
\section{Introduction}

To understand the aims and scope of this paper we must first start in Romania in 1982 with a master student, not unlike myself, who published her master thesis \cite{equivariance} simplifying drastically Yang-Mills equations to a non linear wave equation. In this paper, Mrs. Dumitrascu supposes an 'echivariante' connection and then states that it is of the form:
$$
[A_{\mu}]_{i,j}(x)=g(r) \, (\delta_{\mu}^j\,x_i-\delta_{\mu}^i\,x_j)
$$
without any justification. This has later been used as an ansatz for Yang-Mills in $n$ dimensions on an $SO(n)$ principal bundle in articles such as \cite{Roland} and \cite{Roland2} to study the stability of Yang-Mills solutions with singularities. Here a few questions arise; What kind of connection are we talking about and what does it represent ? Is there a more geometrical or physical meaning behind this ansatz ? Can it be generalized to other groups with more physical meaning perhaps ? Can we find a more general method to simplify such partial differential equations ?

To answer the first question, we will build on works already well-documented in differential geometry in mathematics and gauge theory in physics to define the right objects and the right context in sections 2, 3 and 7. 
In section 2, we define and give intuition for principal fiber bundles on manifolds, while in section 3, we take a look at what a connection or derivative is on such structures. While this might seem rather elementary to some, it is important to not confuse the symmetries that are inherent to those structures with the ones that we will add on. \cite{Geodiff,MathGauge,Pfb}
In section 7,  we will introduce Yang-Mills with its Lagrangian problem to get the needed understanding of the equation. \cite{MathGauge,Elec}

To answer the second question, \cite{Gastel2011YangMillsCO} defines a useful notion of 'equivariant connection' on connections which we will adapt to our context. This will require a group action on our base manifold which we will construct in section 6.
Building upon this, in section 4, we will go further by classifying such equivariant connection and crossing the bridge between the more abstract concept of equivariant connection in \cite{Gastel2011YangMillsCO} to the more concrete ansatz given in \cite{equivariance}, proving their relation.
We will also add physical interpretations to such connections with examples throughout the paper and more specifically in section 9.

To answer the third question, we will use the tools we have built to construct similar results to the ones already given by the ansatz to the groups $SO^+(p,q)$, $Spin^+(p,q)$ and $SU(n)$ and link some of them to fundamental forces of the universe.

All throughout, we will follow a very general method which can be used to simplify other objects in other contexts similarly to \cite{Pal}'s principle of symmetric criticality: 

Supposing we have a map $f:A \mapsto B$ and a Lie group $G$ represented in both $GL(A)$ and $GL(B)$. We will say that $f$ is equivariant if for all $a \in A$ and $g\in G$, $f(g\cdot a)=g\cdot f(a)$. We can then classify such $f$ using that for all $g$ in $a$'s isotropy group, $g$ is in $f(a)$'s isotropy group and that if we know $f(a)$ then we know all $f(g\cdot a)$.

Finally, in sections 8 and 9, we will plug in the equivariant connections we have obtained in section 4 into Yang-Mills with metrics that have similar symmetries to reduce the equations. 
\newpage

\section{Principal Bundles}
In this section we give a brief summary of how to define Principal Bundles. The physical idea behind this object is that we want to encode supplementary information on a manifold $M$. Furthermore, we want to encode the symmetry that this information is acted on by a Lie Group.
\subsection{General Lie Group formalism}
This will only be a quick summary with several big results left out. But if you are a curious we recommend checking out \cite{Geodiff},\cite{Pfb} and \cite{MathGauge}.
\begin{definition}[Lie Group]
A \textbf{Lie Group} $G$ is a group equipped with a manifold structure that allows the following functions to be smooth:
\begin{enumerate}
    \item $L_g:G\longrightarrow G,\, h\mapsto g\cdot h$
    \item $R_g:G\longrightarrow G,\, h\mapsto h\cdot g$
    \item $inv:G\longrightarrow G,\, h\mapsto h^{-1}$
\end{enumerate}
\end{definition}
\begin{definition}[Lie Algebra]
A \textbf{Lie Algebra} $\mathfrak{g}$ is a vector space equipped with an antisymetric bilinear product $[.,.]:\mathfrak{g}\times \mathfrak{g}\longrightarrow \mathfrak{g}$ such that for all $X,Y,Z \in \mathfrak{g}$ we have $[X,[Y,Z]]+[Y,[Z,X]]+[Z,[X,Y]]=0$.
\end{definition}
\begin{example}
$\mathfrak{X}(M)$ is a Lie Algebra if we take a look at the correspondence with derivatives.
\end{example}

We'll now give ourselves a Lie Group $G$.

\begin{definition}[Left Invariant Vector Fields]
Let $X \in \mathfrak{X}(G)$ be a vector field. It is \textbf{Left Invariant} (or in $\mathfrak{X}_L(G)$) if for all $g\in G$
$$
L_g^* X:=T_{\circ L_g}(L_{g^{-1}})\circ X \circ L_g = X
$$
\end{definition}
\begin{proposition}
Since for all $X\in \mathfrak{X}_L(G),\, X(g)=T_e (L_g) X(e)$ we have that $\mathfrak{X}_L (G) \simeq \mathfrak{g}:=T_e G$
\end{proposition}

\begin{definition}[Lie Transformation group]
Given a manifold $M$, $G$ is a \textbf{Lie transformation group} if it has a right action on $M$ such that for all $g\in G,\, R_g:x\in M \mapsto x\cdot g$ is a global diffeomorphism.
\end{definition}

\subsection{Principal Fiber Bundles}
\begin{definition}[Principal Fiber Bundles]
Let $G$ be a Lie group, $M$ and $P$ manifolds with $\pi:P\longrightarrow M$ smooth. The tuple $(P,\pi,M,G)$ is called a \textbf{$G$-principal fiber bundle} over $M$ if
\begin{enumerate}
    \item $G$ is a Lie transformation group on $P$. The action is free and simply transitive on each fibers.
    \item There exists a bundle atlas $\{(U_i,\phi_i)\}$ consisting of $G$-equivariant bundle charts. In other words:
    \begin{enumerate}
        \item The $\{U_i\}$ cover $M$.
        \item $\phi_i:\pi^{-1}(U_i)\longrightarrow U_i \times G$ is a diffeomorphism.
        \item The following diagram commutes:
        \begin{center}
        \begin{tikzcd}
\pi^{-1}(U_i) \arrow[d, "\pi"'] \arrow[r, "\phi_i"] & U_i \times G \arrow[ld, "p_1"] \\
U_i                                                 &                               
\end{tikzcd}
\end{center}
        \item $\phi_i(p\cdot g)=\phi(p)\cdot g$, where $g$ acts on $U_i\times G$ by $(x,a)\cdot g:= (x,a\cdot g)$.
    \end{enumerate}
\end{enumerate}
\end{definition}

Now that we have a space where our supplementary information can live, we need the right functions to get it out of our manifold $M$.

\begin{definition}[Sections]
We call \textbf{section} a smooth function $s:M\longrightarrow P$ such that $\pi \circ s=Id_M$.
\end{definition}

We can now switch points of views and define principal bundles from sections:
\begin{theorem}[Defining a Principal Bundle from sections]
Let $G$ be a Lie group, $\pi:P\longrightarrow M$ a smooth function, then $(P,\pi,M,G)$ is a principal fiber bundle if and only if
\begin{enumerate}
    \item $G$ is a Lie transformation group that acts simply and freely on the fibers.
    \item There exists an open cover $\mathcal{U}=\{U_i\}_{i\in I}$ of $M$ and local sections $s:U_i \longrightarrow P$ for all $i\in I$.
\end{enumerate}
\end{theorem}
\begin{remark}
The bundle atlas functions are defined by the relation 
$$
\begin{aligned}
\phi_i^{-1}=\psi_{s_i}:U_i \times G &\longrightarrow \pi^{-1}(U_i)\\
(x,g) &\longmapsto s_i(x)\cdot g
\end{aligned}
$$
\end{remark}
Let us look at an important example:

\begin{example}[The frame bundle of a manifold]\label{frame bundle}
This example is taken from \cite{MathGauge}.

Let $M$ be an $n$-dimensional manifold then we can set 
$$
GL_x(M):=\{\nu_x=(\nu_1,...,\nu_n)|\nu_x \text{ is a basis of } T_x M\}
$$
Then we define:
$$
GL(M):=\amalg_{x\in M} GL_x(M)
$$
and
$$
\begin{aligned}
\pi:GL(M) &\longrightarrow M \\
\nu_x &\longmapsto x
\end{aligned}
$$
Then $GL_n(\mathbf{R})$ acts on the right with $\nu_x\cdot A=(\sum_{i=1}^n \nu_i \, A_{i,1},...,\sum_{i=1}^n\nu_i\,A_{i,n})$ which is just the left multiplication. 

We then need to get a manifold structure. To do so, take $\{U_i,\phi_i=(x_1,...,x_n)\}$ a bundle chart of $M$. For each $\nu_x\in GL(M)|_{U_i}$ there exists a unique $A(x)\in GL_n(\mathbf{R})$ such that $\nu_x=(e_1,...,e_n)\cdot A(x)$. We then define 
$$
\begin{aligned}
\varphi_i:GL(M)|_{U_i}&\longrightarrow U_i \times GL_n(\mathbf{R}) \\
& \nu_x &\longmapsto (x,A(x))
\end{aligned}
$$
$\varphi_i$ is then bijective and we have the following commuting diagram:
\begin{center}
\begin{tikzcd}
GL(M)|_{U_i} \arrow[dd, "\pi"'] \arrow[rr, "\varphi_i"] &  & U_i \times GL_n(\mathbf{R}) \arrow[lldd, "p_1"] \\
& & \\
U_i                                                    &  &                                               
\end{tikzcd}
\end{center}
We can then verify smoothness to see that we have constructed a $(GL(M),\pi,M,GL_n(\mathbf{R}))$ principal fiber bundle.
\end{example}

Since we will interest ourselves towards Yang-Mills we can look at the frames we are interested in which are the ones invariant by our metric and direct.

\begin{example}[Metric induced frame bundle]
Given $(M,g)$ a semi-Riemannian manifold of signature $(p,q)$ then set
$$
SO_x(M)=\{\nu_x \in GL_x(M) | g_x(\nu_i,\nu_j)=I_{p,q}, \text{ and } det(\nu_x)=1\}
$$
We analoguously still have a principal fiber bundle: $(SO(M),\pi|_{SO(M)},M,SO(p,q))$.
\end{example}

\newpage

\section{Principal Connection and Connection $1$-form}
We can now introduce the notion of a connection on a Principal bundle. We want to be able to transport knowledge about one point on the bundle to another.
\subsection{A geometric distribution}
\begin{definition}[Geometric distribution]
Given a manifold $N$, a \textbf{geometric distribution of dimension $r$} is a map $\mathfrak{E}:x\mapsto E_x \subset T_x N$ such that for all $x\in N$ there exists a neighborhood $U$ and smooth vector fields on $U$, $X_1,...,X_r$ such that $\forall y \in U,\, E_y=span (X_1(y),...,X_r(y))$.
\end{definition}

This is a sort of generalization of a vector field to the $r$-th dimension.

We now place ourselves on a principal bundle $(P,\pi,M,G)$, with the right action defined as above.

\begin{definition}[Vertical Tangent spaces]
The most basic example of a geometric distribution on our bundle is given by taking its fibers $P_x =\pi^{-1}(x)$.

We can now go even further, since they are regular submanifolds of $P$. For any point $u\in P_x$ we can define $Tv_u P := T_u (P_x)$. This is called the vertical tangent space in $u$.
\end{definition}

\begin{proposition}[on Vertical Tangent Spaces]
We have the following properties:
\begin{enumerate}
    \item $Tv_u = ker(T_u \pi )$.
    \item $Tv_u$ is linearly isomorphic to $\mathfrak{g}$ with,
    $$
    \Phi_u :X \mapsto \Tilde{X}(u)=\frac{d}{dt}|_{t=0} (u\cdot (exp(tX)))
    $$
    \item For all $X\in \mathfrak{g}$ we have the following relation with the flow of $\Tilde{X}$,
    $$
    Fl^{\Tilde{X}}_t (u)=u\cdot exp(tX)=R_{exp(tX)} (u)
    $$
\end{enumerate}
\end{proposition}

\begin{definition}[Horizontal Tangent space]
Any subspace $Th_u P\subset T_u P$ that is complimentary to $Tv_u P$ is called an \textbf{Horizontal Tangent space of $u$}.
\end{definition}

\subsection{Principal Connection and Parallel transport}

\begin{definition}[Principal Connection]
A \textbf{Connection} on a principal fiber bundle is a geometric distribution of horizontal spaces, that is right invariant. In other words it is a geometric distribution $Th:P\ni u \mapsto Th_u P \subset T_u P$ such that for all $g\in G$ we have
$$
T_u R_g (Th_u P)=Th_{u\cdot g} P
$$
\end{definition}

This does not look close to the connections we are used to in physics but it is the right notion in this context since it allows us to define parallel transport. But first we need to lift our vector fields in a way that agrees with our connection.

\begin{definition}[Horizontal Lift of a vector field]
Let $X\in \mathfrak{X}(M)$. A vector field $X^*$ is called an \textbf{Horizontal Lift of $X$} if for all $p\in P$ we have:
\begin{enumerate}
    \item $X^* (p) \in Th_p P$, and
    \item $T_p \pi (X^* (p))=X(\pi (p))$
\end{enumerate}
\end{definition}

\begin{theorem}[Existence and unicity of horizontal lifts]
For any $X \in \mathfrak{X}(M)$ there exists a unique horizontal lift $X^* \in \mathfrak{X}(P)$. Moreover $X^*$ is right invariant.
\end{theorem}

\begin{definition}[Horizontal Lift of a path]
A path $\gamma^*:I\longrightarrow P$ is called \textbf{the horizontal lift} of a path $\gamma:I\longrightarrow M$ if 
\begin{enumerate}
    \item $\pi (\gamma^*(t))=\gamma(t)\,\forall t\in I$ and
    \item $\frac{d}{dt} \gamma^*$ is horizontal for all $t\in I$
\end{enumerate}
\end{definition}

\begin{theorem}[Existence and unicity of horizontal lifts]
Let $\gamma:I\longrightarrow M$ be a path in $M$, $t_0 \in I$ and $u\in P_{\gamma (t_0)}$. Then there exists a unique horizontal lift $\gamma_u ^*$ of $\gamma$ such that $\gamma_u ^* (t_0)=u$.
\end{theorem}

We can now define a unique parallel transport associated to our connection.

\begin{definition}[Parallel Transport]
Given a connection $Th$, a path $\gamma: [a,b] \longrightarrow M$ and $u\in P_{\gamma(a)}$ the \textbf{parallel transport of $u$ along $\gamma$ with respect to $Th$} is
$$
P^{Th}_{\gamma} (u)= \gamma^* _u (b)
$$
\end{definition}

Now that we have seen that we have the right notion of connection, let us encode it in a much simpler way.

\subsection{Connection $1$-form}
\begin{definition}[Connection $1$-form]\label{1form}
A \textbf{connection $1$-form} on a principal fiber bundle is a $1$-form $A\in \Omega ^1(P,\mathfrak{g})$ that satisfies
\begin{enumerate}
    \item $R_g ^* A= Ad(g^{-1}) \circ A$ for all $g\in G$ \label{connection}
    \item $A(\Tilde{X})=X$ for all $X \in \mathfrak{g}$
\end{enumerate}
(As a reminder $(R_g ^* A)_u = A_{u\cdot g} \circ T_u R_g$)
\end{definition}

Now we can show that there is a correspondence.

\begin{theorem}[Connection and $1$-form correspondence]
Connections and connection $1$-forms on a principal bundle are in bijective correspondence,
\begin{enumerate}
    \item Given a connection $Th$ then we can define a connection $1$-form by
    $$
    A_u(X(u) \oplus Y_h) := X(u) \, \forall u\in P,\, X\in \mathfrak{g},\, Y_h \in Th_u P
    $$
    \item Given a connection $1$-form $A$ we can define a connection $Th$ by
    $$
    Th_u P:=ker (A_u) \, \forall u\in P
    $$
\end{enumerate}
\end{theorem}

We can finally find an object of the form requested in \cite{equivariance}.

\begin{definition}[The local connection form or Gauge field]
Given a connection form $A\in \Omega ^1 (P,\mathfrak{g})$ and a local section $s:U\longrightarrow P$ we define the \textbf{local connection form or Gauge field}
$$
A^{s} := A_{\circ s} \circ Ts \in \Omega^1 (U,\mathfrak{g})
$$
\end{definition}

To see that this object is useful we'll consider how one can characterize a connection form from local forms.

\begin{theorem}[Local characterization of the connection form]
Given a covering $\{(s_i,U_i)\}$ of $M$ with the connection functions $g_{ij}$ such that $s_i(x)=s_j(x)\cdot g_{ij}(x)$ and $\mu_{ij} \in \Omega^1 (U_i \cap U_j, \mathfrak{g})$ such that $\mu_{ij}(X) = TL_{g_{ij}^{-1}(x)}(Tg_{ij} (X))$ we can state:

If $\{A_i \in \Omega^1 (U_i, \mathfrak{g})\}_i$ is a family of $1$-forms such that whenever $U_i \cap U_j \neq \empty$ we have
$$
A_i= Ad(g_{ij}^{-1}) \circ A_j + \mu_{ij}
$$
then there exist a unique connection form $A \in \Omega^1 (P,\mathfrak{g})$ such that $A^{s_i}=A_i$ for all $i$.
\end{theorem}
\begin{remark}
Notice that since $g_{ij}=p_2 \circ \psi_{s_j}^{-1}\circ s_i$ then $g_{ij}$ is smooth.
\end{remark}

\begin{example}[Bijectivity with local covariant derivatives]
Let us come back on example \ref{frame bundle} with the frame bundle $GL(M)$. If we take $s=(s_1,...,s_n)$ a local section of $GL(M)$ on $U$, which can be seen as a local frame of coordinate. Then the two following applications are inverse of one another and allow for a global bijectivity between connections on the frame bundle and covariant derivatives or connections on a manifold:
$$
\begin{aligned}
\phi_s:\Omega^1(U,\mathfrak{gl}_n(\mathbf{R})) & \longrightarrow \{\text{Covariant Derivatives on } U\} \\
B=\sum_{i,j=1}^n \omega_{i,j}\,e_{i,j} &\longmapsto \nabla \text{ such that } \nabla_X s_k := \sum_{i=1}^n \omega_{i,k}(X)\,s_i
\end{aligned}
$$
$$
\begin{aligned}
\psi_s:\{\text{Covariant Derivatives on } U\} & \longrightarrow \Omega^1(U,\mathfrak{gl}_n(\mathbf{R})) \\
\nabla \text{ such that } \nabla s_i=\sum_{j=1}^n \omega_{j,i} \otimes s_j &\longmapsto \sum_{i,j=1}^n \omega_{i,j}\,e_{i,j}
\end{aligned}
$$
Furthermore, one can show that this bijectivity also holds for between connections on $SO(M)$ and metric induced connections. This motivates us to look local $SO(p,q)$-connections rather than global $SO(g)$ connections.
\end{example}

Now that we have found the objects we're looking for, let's try to characterize them even more in specific situations.

\newpage
\section{Algebraically Equivariant linear $1$-form}

We now need to define what is an equivariant connection. To do so, we will first take a look at the inherent algebraic definition of equivariance.
\begin{definition}[Algebraically Equivariant Linear $1$-form]
Given a manifold $M$ and a matrix Lie group $G$ acting on $M$. We call $\mathfrak{g}$ its Lie algebra and take $H$ a subgroup of $G$. A $1$-form $B \in \Omega^1(M,\mathfrak{g})$ is called \textbf{Algebraically $H$-equivariant} if for all $\Lambda \in H$ and $x\in M$:
$$
B_{\mu}(\Lambda\cdot x)=\sum_{\nu=1}^n Ad(\Lambda)(B_{\nu}(x)) [\Lambda^{-1}]_{\nu,\mu}
$$
\end{definition}

This, as we will see, is closely related to the notion of invariance:
\begin{definition}[Algebraically Invariant Linear form]
In the same context as above, a linear form $f$ is called \textbf{Algebraically $H$-invariant} if for all $\Lambda \in H$:
$$
f_{\mu}=\sum_{\nu=1}^n Ad(\Lambda)(f_{\nu}) [\Lambda^{-1}]_{\nu,\mu}
$$
\end{definition}

Then it seems that the best simplifications we can have are when $H=G$. Let us then compute what such $1$-forms look like.

\begin{proposition}\label{morphisms}
Given a Lie group homomorphism $\lambda: K\longrightarrow G $ with an injective differential $\lambda_d : \mathfrak{k}\longrightarrow \mathfrak{g}$, then $B\in \Omega^1(M,\mathfrak{g})$ is $H$-equivariant, if and only if $\lambda_d^{-1}\circ B \in \Omega^1(M,\mathfrak{k})$ is $\lambda^{-1}(H)$-equivariant with $K$'s action on $M$ defined by $k\cdot x=\lambda(k) \cdot x$.
\end{proposition}
\begin{proof}
Let us take $k\in \lambda^{-1}(H)$. Then we have:
$$
\begin{aligned}
\lambda_{d}^{-1}\circ B(k\cdot x)^{\mu}&=\lambda_{d}^{-1}\circ B(\lambda(k)\cdot x)^{\mu} \\
&= \sum_{\nu=1}^n\lambda_d^{-1}(Ad(\lambda(k))(B(x)))^{\nu} \,[\lambda(k)^{-1}]_{\nu,\mu} \\
&=\sum_{\nu=1}^n Ad(k)(\lambda_d^{-1}\circ B(x))^{\nu} \,[\lambda(k)^{-1}]_{\nu,\mu}
\end{aligned}
$$
Which is what we wanted to show since we can see $K$ as a matrix Lie Group through $\lambda$.
\end{proof}

\begin{remark}
This result allows us to find $H$-equivariant connections by first computing them on $\mathfrak{h}$ before injecting it in $\mathfrak{g}$. 
\end{remark}

A strategy to compute such $1$-forms will come from utilizing the following:

\begin{proposition}[Orbits and Isotropy groups]
If we suppose that $B$ is $H$-equivariant
\begin{enumerate}
    \item If we know the value of $B$ at $x\in M$ then we know the values for the entire orbit of $x$ in $H$.
    \item The linear form $B(x)$ is invariant by the isotropy group of $x$.
\end{enumerate}
\end{proposition}
\begin{proof}
If we suppose $B$ to be $H$-equivariant then:
\begin{itemize}
    \item If $x$ and $y$ are in the same $H$-orbit then there exists $\Lambda \in H$ such that $y=\Lambda \cdot x$ and therefore:
    $$
    B(y)=B(\Lambda\cdot x)=\sum_{\nu=1}^n Ad(\Lambda)(B_{\nu}(x)) [\Lambda^{-1}]_{\nu,\mu}
    $$
    \item Since for all $\Lambda$ in $x$'s isotropy group $\Lambda\cdot x=x$ we have:
    $$
    B(x)=B(\Lambda\cdot x)=\sum_{\nu=1}^n Ad(\Lambda)(B_{\nu}(x)) [\Lambda^{-1}]_{\nu,\mu}
    $$
\end{itemize}
\end{proof}

\subsection{$SO(n)$ simplifications}

We now want to prove that the ansatz in \cite{equivariance} is in fact correct and reduce the problem when $G=SO(n)$.

\begin{lemma}
Given the canonical basis of $\mathbf{R^n}$, $e_1,...,e_n$. $B$ is entirely determined by the function $r\mapsto B(r\, e_1)$.
\end{lemma}
\begin{proof}
Given $x\in U_i$, since we are in $SO(n)$ we can pick $g_x \in SO(n)$ such that $g_x \cdot |x|\, e_1=x$. Using the definition we get
$$
[B_{\mu}(x)]_{i,j}=\sum_{\nu,k,l=1}^n [g_x^{-1}]_{\nu,\mu}\, [g_x]_{i,k}\, [g_x^{-1}]_{l,j} \, [B_{\nu} (|x|\,e_1)]_{k,l} 
$$
and we can see that the coefficients $[B_{\nu} (|x|\,e_1)]_{k,l}$ determine $B(x)$.
\end{proof}

We now need to characterize $B(r\,e_1)$.

Firstly, let us define 
$$\Tilde{SO}(n-1)=\{Diag(1,P)|P\in SO(n-1)\}=\{g\in SO(n)|g\cdot e_1=e_1\}$$
which is isomorphic to $SO(n-1)$. Then we define the action
$$
\begin{aligned}
\rho : SO(n) \times (\mathfrak{so(n)})^n & \longrightarrow (\mathfrak{so(n)})^n \\
(g,(C_{\mu})) & \longmapsto \sum_{\nu=1}^n [g^{-1}]_{\nu,\mu} \, g\cdot C_{\nu} \cdot g^{-1}
\end{aligned}
$$
With that in mind we have,
\begin{lemma}
For all $r\in \mathbf{R}^+,\, B(r\,e_1) \in \bigcap_{g\in \Tilde{SO}(n-1)} ker(\rho (g) -Id) $.

And $[E_{\mu}]_{i,j}=\delta^j_{\mu} \, \delta^i_1 - \delta^i_{\mu} \, \delta^j_1$ is in $\bigcap_{g\in \Tilde{SO}(n-1)} ker(\rho (g) -Id)$.
\end{lemma}
\begin{proof}
Since for all $g\in \Tilde{SO}(n-1),\, g\cdot r\,e_1=r\,e_1$ we have $\rho(g)(B(r\,e_1))=B(g\cdot r\, e_1)=B(r\,e_1)$.

Secondly, if we take $g\in \Tilde{SO}(n-1)$ we have
$$
\begin{aligned}
\relax[\rho(g)(E)_{\mu}]_{i,j}&=\sum_{\nu,k,l=1}^n [g^{-1}]_{\nu,\mu}\, [g]_{i,k}\, [g^{-1}]_{l,j} \, (\delta^l_{\nu} \, \delta^k_1 - \delta^k_{\nu} \, \delta^l_1) \\
&=\sum_{\nu=1}^n [g^{-1}]_{\nu,\mu}\, [g]_{i,1}\,[g^{-1}]_{\nu,j} - [g^{-1}]_{\nu,\mu}\, [g]_{i,\nu} \, [g^{-1}]_{1,j} \\
&= \delta^j_{\mu} \, [g]_{i,1} -\delta^i_{\mu} \, [g]_{j,1}\\
&= \delta^j_{\mu} \, \delta^i_1 - \delta^i_{\mu} \, \delta^j_1 \\
&=[E_{\mu}]_{i,j}
\end{aligned}
$$
Where we used $g^{-1}=g^T$ and the explicit form of matrices in $\Tilde{SO}(n-1)$.
\end{proof}

What we need to do now, is to find out what is $\bigcap_{g\in \Tilde{SO}(n-1)} ker(\rho (g) -Id))$.

Let us take the basis $\{e_{a,b}^{\alpha}\}_{0 \leq a<b \leq n}^{\alpha=1,...,n}$ of $\mathfrak{so(n)}^n$ such that $[e_{a,b}^{\alpha}]_{i,j}^{\mu}=\delta^{\mu}_{\alpha}\,(\delta^i_{a} \, \delta^j_b - \delta^j_{a} \, \delta^i_b)$.
We can do simple computations to see that,
\begin{proposition}
$G$ acts on the basis $\{e_{a,b}^{\alpha}\}_{0 \leq a<b \leq n}^{\alpha=1,...,n}$ with for all $g\in SO(n)$
$$
[\rho (g)(e_{a,b}^{\alpha})]^{\mu}=[g]_{\mu,\alpha}\,g\cdot e^A_{a,b} \cdot g^{-1}
$$
\end{proposition}
One way to represent this is to think of a $n\times n$ matrix of $\frac{n(n-1)}{2}\times \frac{n(n-1)}{2}$ matrices. And if we take $Ad(g)$ and order the basis $e_{a,b}^{\alpha}$ by $\alpha$ first and then by $(i,j)$ we get the matrix of $\rho(g)$ to be
$$
\begin{pmatrix}
g_{1,1}\,Ad(g) & g_{1,2}\,Ad(g) & \cdots & g_{1,n}\,Ad(g) \\
g_{2,1}\,Ad(g) & \ddots &\ddots & \vdots \\
\vdots & \ddots &\ddots & \vdots \\
g_{n,1}\, Ad(g) & \cdots & \cdots & g_{n,n}\, Ad(g)
\end{pmatrix}
= g \otimes Ad(g)
$$

\begin{remark}
With this formalism, $e^{\mu}_{i,j}=e_{\mu} \otimes e^A_{i,j}$.
\end{remark}
\subsubsection{$SO(3)$ case}\label{SO(3)}
we want to concentrate on the $SO(3)$-case to try and get a better look at what direction we can go.

Firstly, a general matrix of $\Tilde{SO}(2)$ looks like 
$$
g^x_{\theta}=
\begin{pmatrix}
1 & 0 & 0 \\
0 & \cos(\theta) & -\sin(\theta) \\
0 & \sin(\theta) & \cos(\theta)
\end{pmatrix}
$$
With further computations in the basis stated above ($e_{1,2},e_{1,3},e_{2,3}$)
$$
Ad(g^x_{\theta})=
\begin{pmatrix}
 \cos(\theta) & -\sin(\theta) & 0 \\
 \sin(\theta) & \cos(\theta) & 0 \\
 0 & 0 & 1
\end{pmatrix}
$$
Both matrices are diagonalizable simultaneously with regards to $\theta$ with eigenvalues $1,e^{i\,\theta},e^{-i\,\theta}$. This means that for all $\theta$, $\rho(g^x_{\theta})=g^x_{\theta}\otimes Ad(g^x_{\theta})$ is diagonalisable with eigenvalues given by multiplying the eignevalues of the two matrices with one another $1,e^{i\,\theta},e^{-i\,\theta},e^{i\,2\,\theta},e^{-i\,2\,\theta}$ but notably, the eigenspace of $1$ is of dimension $3$. This means that $\bigcap_{g\in \Tilde{SO}(2)} ker(\rho (g) -Id))$ is of dimension $3$ and not $1$ as we hoped.

With a bit more work we can find the eigenvectors, 
$$
\begin{aligned}
E_1 &= e_{2,3}^1 + e_{3,1}^2 + e_{1,2}^3 \\
E_2 &= e_{1,2}^2 + e_{1,3}^3 \\
E_3 &= e_{2,3}^1
\end{aligned}
$$

This means that our ansatz should be of the form
$$
B(r\,e_1)=f(r)\,E_1 + g(r)\,E_2 + h(r)\,E_2
$$
and after applying $\rho (g_x)$ to get to $x$ we have the general ansatz for $SO(3)$

\begin{theorem}[$SO(3)$-equivariant $1$-form]
Given an $SO(3)$ $1$-form which is  $SO(3)$-equivariant it is of the form:
$$
    [B_{\mu}(x)]_{i,j}=f(|x|)\,\epsilon_{i,j,\mu} + g(|x|)\,(\delta_j^{\mu}\,x_i\delta_i^{\mu}\,x_j ) + h(|x|)\,x_{\mu}\,\sum_{k=1}^3\epsilon_{i,j,k}\,x_k 
$$
\end{theorem}
\begin{proof}
This is a lot of matrix multiplication and is fully accessible to anyone with enough time.
The basic idea is that $x=|x|\begin{pmatrix}
 \cos(\theta) \, \cos(\phi) \\
 \sin(\theta) \cos(\phi) \\
 \sin(\phi)
\end{pmatrix}
= |x|\, g_{\theta}^z \cdot g_{\phi}^y \cdot e_1$ so we need to apply $\rho(g_{\theta}^z)\cdot \rho(g_{\phi}^y)$ to our eigenvectors to get the wanted result, with a rescale of $|x|$ or $|x|^2$ that goes inside the functions.
\end{proof}
\begin{remark}
The second term is the original ansatz given by \cite{equivariance}.
\end{remark}

\subsubsection{$SO(n)$ case}

How should we proceed ? Well, as we saw we need to characterize the $1$-eigenspace of the generators of $\rho (\Tilde{SO}(n-1))$. To do so we'll need to diagonalize them because as we saw for $n=3$ diagonalizing $g$ seems to allow us to diagonalize $\rho(g)=g\otimes Ad(g)$. 
Let us render that relation explicit.
\begin{proposition}[Simultaneous reduction]\label{simred}
Suppose $g\in SO(n)$ is diagonalizeable (respectively trigonalizeable) with matrix $P$ and eigenvalues $\lambda_1,...,\lambda_n$. Then $g\otimes Ad(g)$ is diagonalizeable (respectively trigonalizeable) with the matrix $P\otimes Tconj(P)$ and eigenvalues $\{\lambda_{\mu} \, \lambda_i \, \lambda_j \}_{i<j} ^{\mu}$.

In other words if $v_1,...,v_n$ are eigenvectors of $g$ then $\{v_{\mu} \otimes (v_i \cdot v_j^T - v_j \cdot v_i^T)\}_{i<j}^{\mu}$ are eigenvectors of $\rho(g)$ with eigenvalues as multiples of the original..

Here $Tconj(P)(M)= P\cdot M \cdot P^T$.
\end{proposition}
\begin{proof}
Notice that our basis of $\mathfrak{so(n)}$ is given by the basis of $\mathbf{R}^n$ with $e^A_{i,j}=e_i \cdot e_j^T - e_j \cdot e_i^T$.

We already know that $P\otimes Tconj(P)$ is an invertible matrix, we just need to show it diagonalizes $g$. To do so let us compute:
$$
\begin{aligned}
(g\otimes Ad(g)) \cdot (P\otimes Tconj(P)) \cdot & (e_{\mu}\otimes e^A_{i,j}) = [(g\cdot P) \otimes Tconj(g\cdot P) ] \cdot (e_{\mu}\otimes e^A_{i,j}) \\
&= (g\cdot P \cdot e_{\mu})\otimes (g\cdot P \cdot (e_i \cdot e_j^T - e_j \cdot e_i^T) \cdot (g \cdot P)^T \\
&= \lambda_{\mu} \, (P\cdot e_{\mu})\otimes ( (g\cdot P \cdot e_i)\cdot (g\cdot P \cdot e_j)^T -(g\cdot P \cdot e_j)\cdot (g\cdot P \cdot e_i)^T ) \\
&= \lambda_{\mu}\, \lambda_i \, \lambda_j \, (P\cdot e_{\mu})\otimes ( ( P \cdot e_i)\cdot ( P \cdot e_j)^T -( P \cdot e_j)\cdot ( P \cdot e_i)^T ) \\
&= \lambda_{\mu}\, \lambda_i \, \lambda_j \, (P\otimes Tconj(P))\cdot (e_{\mu} \otimes e^A_{i,j})
\end{aligned}
$$
Which is exactly what we wanted to show.
\end{proof}

\begin{remark}
Since, for $2 \leq i<j$ and $\theta \in [0,2 \pi[$ $g_{i,j}^{\theta}$ are the standard group elements of $\Tilde{SO}(n-1)$ we can first describe $ker(\rho(g_{i,j}^{\theta})-Id)$ and take the intersection before checking that it works with all elements of $SO(n)$.
\end{remark}

\begin{lemma}[Diagonalization of $g_{i,j}^{\theta}$]
The element $g_{i,j}^{\theta}$ has the following eigenvalues with eigenvectors:
\begin{enumerate}
    \item $1$ with $\{e_k\}_{k \neq i,j}$.
    \item $e^{i\,\theta}$ with $e_i + i\,e_j$.
    \item $e^{-i\,\theta}$ with $e_i - i\,e_j$.
\end{enumerate}
\end{lemma}
\begin{proof}
This is left as a good exercise of linear algebra.
\end{proof}
We can then use proposition \ref{simred} to get the following.
\begin{proposition}[$1$-eigenspace of $\rho(g_{i,j}^{\theta})$]\label{1eigen}
The $1$-eigenspace of $\rho(g_{i,j}^{\theta})$ with $\theta \neq 0$ is 
$$
\begin{aligned}
E_{i,j}=span\{ &(e_k \otimes e^A_{l,m})^{\{k,l,m\}\cap \{i,j\}= \emptyset }_{l<m} , \, (e_i \otimes e^A_{i,k} + e_j \otimes e^A_{j,k})_{k\notin \{i,j\}}, \\
& (e_i \otimes e^A_{j,k} - e_j \otimes e^A_{i,k})_{k\notin \{i,j\}},\, (e_k\otimes e^A_{i,j} )_{k \notin \{i,j\}} \} 
\end{aligned}
$$
\end{proposition}
\begin{proof}
There are $4$ ways to get $1$ by multiplying three of $g_{i,j}^{\theta}$'s eigenvalues with interchangeability on the two last coefficients. Thanks to proposition \ref{simred} this gives us the eigenvectors of $\rho(g_{i,j}^{\theta})$.
\begin{enumerate}
    \item $1\times 1 \times 1$
    
    This gives us $(e_k \otimes e^A_{l,m})^{\{k,l,m\}\cap \{i,j\}= \emptyset }_{l<m}$
    \item $e^{i\, \theta} \times e^{-i\, \theta} \times 1$
    
    This gives us $\{ (e_i+i\,e_j)\otimes (e^A_{i,k} - i\,e^A_{j,k}) \}_{k\notin \{i,j\}}$
    \item $e^{-i\, \theta} \times e^{i\, \theta} \times 1$
    
    This gives us $\{ (e_i-i\,e_j)\otimes (e^A_{i,k} + i\,e^A_{j,k}) \}_{k\notin \{i,j\}}$
    \item $1\times e^{i\,\theta} \times e^{-i\,\theta}$
    
    this gives us $\{-2\,i\,  (e_k \otimes e^A_{i,j})\}_{k\notin \{i,j\}}$
\end{enumerate}
Finally, we independently add and subtract the second and third basis and divide by any complex number to get the announced basis.
\end{proof}

\begin{remark}
This coincides with what we've found in the $SO(3)$ case \ref{SO(3)}.
\end{remark}
We then just need to intersect, and we'll find that, unlike in \cite{equivariance}, this gives us different cases.

\begin{theorem}[The different $SO(n)$-equivariance cases]\label{intersecteigen}
We have the following:
\begin{enumerate}
    \item For $n=3$,
    $$
    \bigcap_{2\leq i <j}E_{i,j}=E_{2,3}=span\{e_2 \otimes e^A_{2,1} + e_3 \otimes e^A_{3,1}, \, e_2\otimes e^A_{3,1} - e_3 \otimes e^A_{2,1} , \, e_1 \otimes e^A_{2,3} \}
    $$
    \item For $n=4$,
    $$
    \bigcap_{2\leq i <j}E_{i,j}=E_{2,3}\cap E_{2,4} \cap E_{3,4} = span \{ e_2 \otimes e^A_{2,1} + e_3 \otimes e^A_{3,1} + e_4 \otimes e^A_{4,1}, \, e_2 \otimes e^A_{3,4} - e_3 \otimes e^A_{2,4} + e_4\otimes e^A_{2,3} \}
    $$
    \item For $n\geq 5$,
    $$
    \bigcap_{2\leq i <j}E_{i,j}= span\{ \sum_{k=1}^n e_k \otimes e^A_{k,1} \}
    $$
\end{enumerate}
\end{theorem}
\begin{remark}
We can notice that for $n\geq 5$ we get a one dimensional results which indicates that this must be the result of \cite{equivariance}. However, we must not forget the specific cases for $n=3,\,4$.
\end{remark}
To prove this theorem we'll need the following lemma.
\begin{lemma}\label{intersectlemma}
For $i<j<n$ we have 
$$
\begin{aligned}
E_{i,j} \cap E_{i,j+1}=span \{& (e_k \otimes e^A_{l,m} )_{l<m}^{\{k,l,m\}\cap \{i,j,j+1\} =\emptyset},\,(e_i \otimes e^A_{i,k} + e_j \otimes e^A_{j,k} + e^A_{j+1} \otimes e^A_{j+1,k})_{k \notin \{i,j,j+1\}} , \\
& \, e_i \otimes e^A_{j,j+1} - e_j \otimes e^A_{i,j+1} + e_{j+1} \otimes e^A_{i,j} \}
\end{aligned}
$$
\end{lemma}
\begin{proof}
Given $x\in E_{i,j} \cap E_{i,j+1}$ let us decompose it on the basis we found in proposition \ref{1eigen}.
$$
\begin{aligned}
x &= \sum_{k \notin \{i,j\}} [\sum_{\substack{l<m \\ \{l,m\}\cap \{i,j\}=\emptyset}} x_{k,l,m}e_k \otimes e^A_{l,m} + x^1_k\, (e_i\otimes e^A_{i,k} + e_j\otimes e^A_{j,k}) \\
& + x^2_k\, (e_i\otimes e^A_{j,k} - e_j \otimes e^A_{i,k}) + x^3_k\, (e_k\otimes e^A_{i,j})]\\
&= \sum_{k \notin \{i,j+1\}} [\sum_{\substack{l<m \\ \{l,m\}\cap \{i,j+1\}=\emptyset}} x'_{k,l,m}e_k \otimes e^A_{l,m} + (x')^1_k\, (e_i\otimes e^A_{i,k} + e_{j+1}\otimes e^A_{j+1,k}) \\
& + (x')^2_k\, (e_i\otimes e^A_{j+1,k} - e_{j+1} \otimes e^A_{i,k}) + (x')^3_k\, (e_k\otimes e^A_{i,j+1})]
\end{aligned}
$$
Now we just need to project on the basis element to get the following relations, for $k,\,l,\,m\notin \{i,\,j,\,j+1\}$,
\begin{itemize}
    \item $(k,l,m),\,x_{k,l,m}=x'_{k,l,m}$
    \item $(i,i,j+1),\,x^1_{j+1}=0$
    \item $(i,i,k)$ and $(j+1,j+1,k) ,\,x^1_k=(x')^1_k=x_{j+1,j+1,k}$
    \item $(j+1,j+1,j),\, x_{j+1,j+1,j}=0$
    \item $(i,j,k),\, x^2_k=0$
    \item $(i,j,j+1)$ and $(j+1,i,j),\,x^2_{j+1}=-(x')^2_j=x^3_{j+1}$
    \item $(k,i,j),\,x^3_k=0$
\end{itemize}
In the end we get that $x$ is of the form,
$$
\begin{aligned}
x &= \sum_{k \notin \{i,j,j+1\}} [\sum_{\substack{l<m \\ \{l,m\}\cap \{i,j\}=\emptyset}} x_{k,l,m}\,e_k \otimes e^A_{l,m} + x^1_k\, (e_i\otimes e^A_{i,k} + e_j\otimes e^A_{j,k}+ e_{j+1} \otimes e^A_{j+1,k})] \\
& + x^2\, (e_i\otimes e^A_{j,j+1} - e_j \otimes e^A_{i,j+1} + e_{j+1} \otimes e^A_{i,j})
\end{aligned}
$$
One can then check that all the vectors in the basis are in the intersection.
\end{proof}

We can now start to prove the theorem.

\begin{proof}

\begin{enumerate}
    \item $n=3$
    
    This is immediate and coincides with what we did in the \ref{SO(3)}.
    
    \item $n=4$
    
    Notice that $\bigcap_{2\leq i <j}E_{i,j}=E_{2,3}\cap E_{2,4} \cap E_{3,4} = (E_{2,3} \cap E_{2,4}) \cap (E_{4,2} \cap E_{4,3}) $. Now if we use lemma \ref{intersectlemma} we get
    $$
    E_{2,3}\cap E_{2,4}=span \{ e_2 \otimes e^A_{2,1} + e_3 \otimes e^A_{3,1} + e_4 \otimes e^A_{4,1}, \, e_2 \otimes e^A_{3,4} - e_3 \otimes e^A_{2,4} + e_4\otimes e^A_{2,3} \}
    $$
    and 
    $$
    E_{4,2} \cap E_{4,3}=span \{ e_2 \otimes e^A_{2,1} + e_3 \otimes e^A_{3,1} + e_4 \otimes e^A_{4,1}, \, e_4 \otimes e^A_{2,3} - e_2 \otimes e^A_{4,3} + e_3\otimes e^A_{4,2} \}
    $$
    Which are equal (using the identity $e^A_{i,j}=-e^A_{j,i}$).
    
    \item $n\geq 5$
    
    Notice that $\bigcap_{2\leq i <j}E_{i,j}=\bigcap_{i=2}^{n-1} \bigcap_{j=i+1}^n E_{i,j}$.
    Fix $i$ and let us first prove a result by induction for $n'-(i+1)\geq 2$
    $$
    \begin{aligned}
    \mathcal{P}(n')\equiv [\bigcap_{j=i+1}^{n'} E_{i,j} = span\{&(e_k\otimes e^A_{l,m})^{\{k,l,m\}\cap \{i,...,n'\}=\emptyset}_{l<m},\\
    & (\sum_{l=i}^{n'} e_l \otimes e^A_{l,k})_{k\notin \{i,...,n'\}} \} ]
    \end{aligned}
    $$
    \begin{itemize}
        \item \textit{Base case} $n'=i+3$
        
        Let us take $x\in E_{i,i+1}\cap E_{i,i+2}\cap E_{i,i+3}$; Notice that $E_{i,i+1}\cap E_{i,i+2}\cap E_{i,i+3}=(E_{i,i+1}\cap E_{i,i+2})\cap (E_{i,i+2}\cap E_{i,i+3})$. Using the lemma \ref{intersectlemma} we can write $x$ as
        $$
        \begin{aligned}
        x &= \sum_{k \notin \{i,i+1,i+2\}} [\sum_{\substack{l<m \\ \{l,m\}\cap \{i,i+1\}=\emptyset}} x_{k,l,m}\,e_k \otimes e^A_{l,m} + x^1_k\, (e_i\otimes e^A_{i,k} + e^A_{i+1}\otimes e^A_{i+1,k}+ e_{i+2} \otimes e^A_{i+2,k})] \\
        & + x^2\, (e_i\otimes e^A_{i+1,i+2} - e_{i+1} \otimes e^A_{i,i+2} + e_{i+2} \otimes e^A_{i,i+1})\\
        &=\sum_{k \notin \{i,i+2,i+3\}} [\sum_{\substack{l<m \\ \{l,m\}\cap \{i,i+2\}=\emptyset}} x'_{k,l,m}\,e_k \otimes e^A_{l,m} + (x')^1_k\, (e_i\otimes e^A_{i,k} + e_{i+2}\otimes e^A_{i+2,k}+ e_{i+3} \otimes e^A_{i+3,k})] \\
        & + (x')^2\, (e_i\otimes e^A_{i+2,i+3} - e_{i+2} \otimes e^A_{i,i+3} + e_{i+3} \otimes e^A_{i,i+2})
        \end{aligned}
        $$
        Then we project to get the following relations for $\{k,\,l,\,m\}\cap \{i,\,i+1,\,i+2,\,i+3\}=\emptyset$
        \begin{itemize}
            \item $(k,l,m),\,x_{k,l,m}=x'_{k,l,m}$
            \item $(i,i,k)$ and $(i+3,i+3,k),\,x^1_k=(x')^1_k=x_{i+2,i+2,k}$
            \item $(i,i+1,i+2),\,x^2=0$
        \end{itemize}
        So we get $x$ of the form:
        $$
        \begin{aligned}
        x=\sum_{k \notin \{i,i+1,i+2,i+3\}} [&\sum_{\substack{l<m \\ \{l,m\} \cap \{i,i+1,i+2,i+3\}=\emptyset}} x_{k,l,m}\, e_k \otimes e^A_{l,m} + \\
        & x^1_k\, (e_i\otimes e^A_{i,k} + e_{i+1}\otimes e^A_{i+1,k}+ e_{i+2} \otimes e^A_{i+2,k} + e_{i+3}\otimes e^A_{i+3,k})]
        \end{aligned}
        $$
        Which proves the base case.
        
        \item \textit{Induction step}
        
        Let us suppose $\mathcal{P}(n')$ and prove $\mathcal{P}(n'+1)$.
        
        Take $x \in \bigcap_{j=i+1}^{n'+1} E_{i,j}=(\bigcap_{j=i+1}^{n'} E_{i,j}) \cap E_{i,n'+1}$. Using the induction hypothesis and proposition \ref{1eigen} we can write two ways. Using the same argument as above we get
        $$
        x=\sum_{k\notin \{i,...,n'+1\}} [\sum_{\substack{\{l,m\}\cap \{i,...,n'+1\}=\emptyset \\ l<m}} x_{k,l,m} \, e_k \otimes e^A_{l,m} + x^1_k\, (\sum_{l=i}^{n'+1} e_l \otimes e^A_{l,k})]
        $$
        Which proves the induction step and therefore the claim.
    \end{itemize}
    
    Now that we have for all $i>2$
    $$
    F_i=\bigcap_{j=i+1}^{n} E_{i,j}=span\{(e_k\otimes e^A_{l,m})^{k<i}_{l<m<i},\, (\sum_{l=i}^{n} e_l \otimes e^A_{l,k})_{k<i} \}
    $$
    We can notice that the $F_i$ are increasingly embedded into each other and that therefore $\bigcap_{i=2}^n F_i=F_2=span\{\sum_{l=2}^{n}e_l\otimes e^A_{l,1}\}$.
\end{enumerate}
\end{proof}

Now we can use theorem \ref{intersecteigen} to characterize the $SO(n)$-equivariant connection and give a similar result to \cite{equivariance}.

\begin{theorem}[$SO(n)$ $1$-form]\label{SO(n)-reduction}
Given a $SO(n)$ $1$-form which is $SO(n)$-equivariant then it is of the form:
\begin{enumerate}
    \item For $n=3$
    $$[B_{\mu}(x)]_{i,j}=f(|x|)\,\epsilon_{i,j,\mu} + g(|x|)\,( \delta_j^{\mu}x_i-\delta_i^{\mu}\,x_j ) + h(|x|)\,x_{\mu}\,\sum_{k=1}^3\epsilon_{i,j,k}\,x_k 
    $$
    \item For $n=4$
    $$
    [B_{\mu}(x)]_{i,j}= f(|x|) \, \sum_{k=1}^{4} \epsilon_{k,i,j,\mu} \, x_k + g(|x|)\,( \delta_j^{\mu}\,x_i-\delta_i^{\mu}\,x_j )
    $$
    \item For $n\geq 5$
    $$[B_{\mu}(x)]_{i,j}=g(|x|)\,( \delta_j^{\mu}\,x_i-\delta_i^{\mu}\,x_j ) 
    $$
\end{enumerate}
With smooth functions $f,\,g,\,h:\mathbf{R}\longrightarrow \mathbf{R}$.
\end{theorem}

\begin{proof}
One can check that the intersection of the eigenbasis do indeed give vectors which are invariant by $\Tilde{SO}(n-1)$. And thus we get: 
\begin{enumerate}
    \item For $n=3$
    $$
    B(r\,e_1)= f(r)\,(e_2 \otimes e^A_{2,1} + e_3 \otimes e^A_{3,1}) + g(r)\, (e_2\otimes e^A_{3,1} - e_3 \otimes e^A_{2,1}) + h(r) \, (e_1 \otimes e^A_{2,3})
    $$
    \item For $n=4$
    $$
    B(r\,e_1)= g(r)\,(e_2 \otimes e^A_{2,1} + e_3 \otimes e^A_{3,1} + e_4 \otimes e^A_{4,1}) +f(r)\,( e_2 \otimes e^A_{3,4} - e_3 \otimes e^A_{2,4} + e_4\otimes e^A_{2,3})
    $$
    \item For $n\geq 5$
    $$
    B(r\,e_1)= g(r)\,\sum_{k=1}^n e_k \otimes e^A_{k,1}
    $$
\end{enumerate}
All we have to do now is to apply $\rho(g_x)=g_x \otimes Ad(g_x)$ to the eigenbasis we've found in theorem \ref{intersecteigen} since $\rho(g_x)(B(|x|\,e_1)=B(|x|\,g_x \cdot e_1)=B(x)$.
\begin{enumerate}
    \item For $n=3$, we've already done this in the $SO(3)$ section.
    \item For $n=4$, we will only do the first vector as the second one is the same for all cases.
    
    We'll suppose $x_1 \neq 0$
    Firstly, we will equip ourselves with a matrix $g_x$. But how can we find one ?
    Well, since $g_x \in SO(n) \subset O(n) $ we can at least start with a Gram-Schmidt's orthonormalisation of the basis $(x,\,e_2,\,e_3,\,e_4)$. This gives us 
    $$
    \begin{aligned}
    \Tilde{g}_x=(& \frac{x}{|x|},\,\frac{1}{|x|\,\sqrt{|x|^2 -x_2^2}}\,\begin{pmatrix}
     -x_2\,x_1 \\
     |x|^2-x_2^2 \\
     -x_2\,x_3 \\
     -x_2\,x_4\\
    \end{pmatrix} ,\\
    & \frac{1}{\sqrt{|x|^2-x_2^2}\,\sqrt{x_1^2+x_4^2}}\, \begin{pmatrix}
     -x_3\,x_1 \\
     0 \\
     x_1^2 + x_4^2 \\
     -x_3\,x_4
    \end{pmatrix} ,\, \frac{1}{|x_1|\,\sqrt{x_1^2+x_4^2}}\,\begin{pmatrix}
     -x_4\,x_1 \\
     0 \\
     0 \\
     x_1^2
    \end{pmatrix} )
    \end{aligned}
    $$
    Then one can compute
    $$
    \rho(\Tilde{g}_x)(e_2 \otimes e^A_{3,4} - e_3 \otimes e^A_{2,4} + e_4\otimes e^A_{2,3})_{i,j}^{\mu} = \pm \frac{1}{|x|}\, \sum_{k=1}^4 \epsilon_{k,i,j,\mu} \, x_k
    $$
    With the sign depending on the sign of $x_1$. Furthermore, if we change the sign of the last column we still have $\Tilde{g}_x\cdot|x|\,e_1=x$ and the same result but with the sign reversed but the determinant of $\Tilde{g}_x$ changed. We can therefore pick $g_x \in SO(n)$ such that the result holds and thus the formula is true.
    \item For any $n$, we need to compute, $\rho(g_x)(\sum_{k=1}^n e_k \otimes e^A_{k,1})$.
    
    Firstly, notice that $[\sum_{k=1}^n e_k \otimes e^A_{k,1}]_{i,j}^{\mu} = \delta_j^{\mu} \, \delta_i^1 - \delta_i^{\mu} \, \delta_j^1$
    Then let us compute,
    $$
    \begin{aligned}
    \rho(g_x)(\sum_{k=1}^n e_k \otimes e^A_{k,1}) & = \sum_{\nu,k,l=1}^n [g_x^{-1}]_{\nu,\mu}\,[g_x]_{i,k} \, [g_x^{-1}]_{l,j}\,(\delta_l^{\nu} \, \delta_k^1 - \delta_k^{\nu} \, \delta_l^1) \\
    & = [g_x]_{i,1}\,\sum_{\nu=1}^n [g_x^{-1}]_{\nu,\mu} \, [g_x^{-1}]_{\nu,j} - [g_x^{-1}]_{1,j}\, \sum_{\nu=1}^n [g_x^{-1}]_{\nu,\mu} \, [g_x]_{i,\nu} \\
    & = \frac{1}{|x|}\,(\delta_j^{\mu}\,x_i - \delta_i^{\mu}\,x_j)
    \end{aligned}
    $$
    Which proves the theorem.
\end{enumerate}
\end{proof}

\subsection{$SO^+(p,q)$ simplifications}

As a lot of the following computations are similar to what we have done before, we will pass over them a bit faster.

We will now consider a metric $||.||$ on $\mathbf{R}^{p+q}$ encoded by the symmetric matrix $I_{p,q}=\begin{pmatrix}
 I_p & \\
 & -I_q \\
\end{pmatrix}$.

We will suppose without loss of generality that $1\leq p\leq q$.
If we introduce the Lie Group: $SO(p,q)=\{g\in GL_{p+q}(\mathbf{R})|A^T \cdot I_{p,q} \cdot A= I_{p,q} \} \cap ker(det)$. Then we want to define time preserving matrices. To do so we will decompose $\Lambda \in SO(p,q)$ by blocks of size $p\times p,\, p\times q,\,q\times p$ and $q\times q$:
$$
\begin{pmatrix}
 \Lambda_{1,1} & \Lambda_{1,2} \\
 \Lambda_{2,1} & \Lambda_{2,2}
\end{pmatrix}
$$
Then $\Lambda\in SO^+(p,q)$ if and only if $det(\Lambda_{1,1})>0$, which means it preserves the orientability of both space and time.

To allow such simplifications, we will first take the (locally) time like vectors $M_{p,q}=\{v=(t,x)\in \mathbf{R}^{p+q} | v^T\cdot I_{p,q} \cdot v>0 \}$ and we have the following properties.

\subsubsection{General properties}
This subsection sometimes refers to some results found in \cite{Compute} and \cite{MathGauge}.

\begin{proposition}\label{p,q1stred}
For $v\in M_{p,q}$ there exists $g_v \in SO^+(p,q)$ such that $g_v\cdot ||v||\,e_1=v$
\end{proposition}
\begin{proof}
This is a constructive proof using Gram-Schmidt.
Let us take $(v_1=v, v_2,...,v_n)$ a basis of vectors, where the first $p$ are in $M_{p,q}$ and the following $q$ are space-like ($v_k^T \cdot I_{p,q} \cdot v_k <0$). Then construct
$$
\begin{cases}
y_1=\frac{v_1}{||v_1||} \\
y_{k+1}=\frac{y^*_{k+1}}{||y^*_{k+1}||} & y^*_{k+1}=v_{k+1}-\sum_{l=1}^{p} (y_l^T \cdot I_{p,q}\cdot  v_{k+1})\,y_{l} +\sum_{l=p+1}^{k} (y_l^T \cdot I_{p,q}\cdot  v_{k+1})\,y_{l} 
\end{cases}
$$
Notice that this is well defined because we have the induction case:

$\{y_l\}_{l<k+1}$ is free for all $k$ because it is a basis of $span\{v_1,...,v_k\}$.

$$
\forall l<k+1,\, y_l^T\cdot I_{p,q} \cdot y^*_{k+1}=y_l^T\cdot I_{p,q} \cdot v_{k+1} -y_l^T\cdot I_{p,q} \cdot v_{k+1} \times 1 =0
$$
and, when $k< p$
$$
\begin{aligned}
||y^*_{k+1}||^2 &=(v_{k+1}-\sum_{l=1}^{k} (y_l^T \cdot I_{p,q}\cdot  v_{k+1})\,y_{l})^T\cdot I_{p,q} \cdot (v_{k+1}-\sum_{l=1}^{k} (y_l^T \cdot I_{p,q}\cdot  v_{k+1})\,y_{l}) \\
&= ||v_{k+1}||^2 + \sum_{l=1}^k |y_l^T \cdot I_{p,q}\cdot  v_{k+1}|^2 -2\,\sum_{l=1}^k |y_l^T \cdot I_{p,q}\cdot  v_{k+1}|^2 \\
&=||v_{k+1}||^2 - \sum_{l=1}^k |y_l^T \cdot I_{p,q}\cdot  v_{k+1}|^2 \\
&
>0 
\end{aligned}
$$
Where the last inequality is given by Cauchy Schwarz (which works in this case thanks to $I_{p,q}$ being positive definite on $M_{p,q}$).

When $k\geq p$
$$
\begin{aligned}
(y^*_{k+1})^T\cdot I_{p,q} \cdot y^*_{k+1}&=(v_{k+1})^T\cdot I_{p,q} \cdot v_{k+1} - \sum_{l=1}^p |y_l^T \cdot I_{p,q}\cdot  v_{k+1}|^2 +\sum_{l=p+1}^k |y_l^T \cdot I_{p,q}\cdot  v_{k+1}|^2 \\
&\leq (v_{k+1})^T\cdot I_{p,q} \cdot v_{k+1} +\sum_{l=p+1}^k |y_l^T \cdot I_{p,q}\cdot  v_{k+1}|^2 \\
&<0
\end{aligned}
$$
Where the last inequality is given by the reverse Cauchy Schwarz (which works in this case thanks to $I_{p,q}$ being negative definite for vectors with negative metric).
We can then make a matrix $A=(y_1,...,y_n)$ and notice that $A^T \cdot I_{p,q} \cdot A= I_{p,q}$ and $A\cdot e_1=y_1=\frac{v}{||v||}$. We now have $A \in O(p,q)$ such that $A\cdot ||v|| \, e_1=v$. Since $det(A)=\pm 1$, to get $A\in SO(p,q)$ we can just multiply the last column by $-1$. After this we can just multiply both a column before $p$ and after $p$ by $-1$ to get to be in $SO^+(p,q)$.
\end{proof}

We will need to characterise the Lie algebra.

\begin{proposition}[$\mathfrak{so}(p,q)$]
We have $\mathfrak{so}(p,q)=\{M|M^T\cdot I_{p,q}+ I_{p,q} \cdot M=0,\, Trace(M)=0\}$.
\end{proposition}
\begin{proof}
Firstly, given $\gamma: I \longrightarrow SO^+(p,q)$ smooth such that $\gamma(0)=Id$ then by deriving the defining equation of $O(p,q)$ in $0$ we get
$$
\gamma'(0)^T\cdot I_{p,q}+ I_{p,q} \cdot \gamma'(0)=0
$$
and since $det(\gamma(t))=1$ for all $t$, we can derive to get $Trace(\gamma'(0))=0$.

Conversely, take such a matrix $M$. Then we need to check that $\gamma(t)=exp(t\,M)\in SO(p,q)$ for all $t$.
This is due to $det(exp(t\,M))=exp(Trace(t\,M))=1$ and 
$$
I_{p,q} \cdot exp(t\,M)^T \cdot I_{p,q}=I_{p,q} \cdot exp(t\,M^T) \cdot I_{p,q} =exp(t\,I_{p,q}\cdot M^T \cdot I_{p,q})=exp(-t\,M)=(exp(t\,M))^{-1}
$$
\end{proof}

Using \cite{Compute} we can get an even better characterisation and a basis.

\begin{proposition}
We have 
$$
\begin{aligned}
\mathfrak{so}(p,q) &=\{\begin{pmatrix}
X_1 & X_2 \\
X_2^T & X_3
\end{pmatrix}| X_1\in \mathfrak{so}(p),\, X_3 \in \mathfrak{so}(q),\, X_2 \in M_{p,q}(\mathbf{R})\}\\
&=span\{(e^A_{i,j}|1\leq i<j \leq p\, or \, p+1\leq i<j \leq p+q),\, (e^S_{i,j}| 1\leq i \leq p,\, p+1 \leq j \leq p+q)\}
\end{aligned}
$$
And so $dim(\mathfrak{so}(p,q))=\frac{n\, (n-1)}{2}$
\end{proposition}

This gives rises to the following standard group elements of $SO^+(p,q)$:
\begin{definition}[Standard group elements of $SO^+(p,q)$]
The standard elements of $SO^+(p,q)$ are:

\begin{itemize}
    \item For $1\leq i<j \leq p$ or $p+1\leq i<j \leq p+q$
    
    $$g_{i,j}^{\theta} = \begin{pmatrix}
I_{i-1} & & \\
 & \begin{matrix}
  \cos{\theta} & & -\sin{\theta} \\
   & I_{j-i-1} & \\
   \sin{\theta} & & \cos{\theta} 
 \end{matrix} & \\
 & & I_{n-j}
\end{pmatrix}$$
\item For $1\leq i \leq p,\, p+1 \leq j \leq p+q$ 
$$
k_{i,j}^h= \begin{pmatrix}
I_{i-1} & & \\
 & \begin{matrix}
  \cosh{h} & & \sinh{h} \\
   & I_{j-i-1} & \\
   \sinh{h} & & \cosh{h} 
 \end{matrix} & \\
 & & I_{n-j}
\end{pmatrix}
$$
\end{itemize}
\end{definition}
And we can find their eigenvalues and eigenvectors:
\begin{proposition}[Diagonalisation of $SO^+(p,q)$'s standard elements]
$SO^+(p,q)$'s standard elements have the following eigenvalues and eigenvectors.

\begin{table}[h!]
\caption{Eigenvalues and eigenvectors of $SO^+(p,q)$'s standard elements}
\begin{tabular}{|c||c|c|c|}
\hline
Eigenvalues & $1$ & $e^{i\,\theta}$ & $e^{-i\,\theta}$ \\
\hline
\hline
$g_{i,j}^{\theta}$ & $(e_{\alpha})_{\alpha \neq i,j}$ & $e_i + i\, e_j$ & $e_i-i\,e_j$ \\
\hline
\end{tabular}
\begin{tabular}{|c||c|c|c|}
\hline
Eigenvalues & $1$ & $e^h$ & $e^{-h}$ \\
\hline
\hline
$k_{i,j}^{h}$ & $(e_{\alpha})_{\alpha \neq i,j}$ & $e_i + e_j$ & $e_i-e_j$ \\
\hline
\end{tabular}
\end{table}
\end{proposition}

\subsubsection{$SO^+(p,q)$ case}
This specific case is very motivated by physics as $SO^+(p,q)$-equivariance could signify that symmetry on the metric in the base manifold is reflected on the connection.
\begin{remark}
In our case $\rho(g)=(I_{p,q}\cdot g \cdot I_{p,q}) \otimes Ad(g)$ where $Ad(g)(M)=g\cdot M \cdot I_{p,q} \cdot g^T \cdot I_{p,q}$
\end{remark}
With that in mind, we can construct an eigenbasis in $\mathbf{R}^{p,q} \otimes \mathfrak{so}(p,q)$ from the one we have in $\mathbf{R}^{p,q}$.

\begin{proposition}[Simultaneous diagonalisation in $SO^+(p,q)$]
Given $g\in SO^+(p,q)$ and an eigenbasis $v_1,...,v_n$ with eigenvectors $\lambda_1,...,\lambda_n$. Then an eigenbasis of $\rho(g)$ in $\mathbf{R}^n \otimes \mathfrak{so}(p,q)$ is given by the eigenvectors for $\alpha$ and $i<j$:
$$
(I_{p,q}\cdot v_{\alpha}) \otimes ((v_i\cdot v_j^T-v_i\cdot v_j^T) \cdot I_{p,q})
$$
With eigenvalue $\lambda_{\alpha} \, \lambda_i \, \lambda_j$.
\end{proposition}
\begin{proof}
Firstly, notice that $((v_i\cdot v_j^T-v_i\cdot v_j^T) \cdot I_{p,q})^T\cdot I_{p,q} + I_{p,q} \cdot ((v_i\cdot v_j^T-v_i\cdot v_j^T) \cdot I_{p,q})^T=0$. Therefore the vectors are in $\mathbf{R}^n \otimes \mathfrak{so}(p,q)$.

Secondly,
$$
\begin{aligned}
\rho(g)((I_{p,q}\cdot v_{\alpha}) \otimes ((v_i\cdot v_j^T-v_i\cdot v_j^T) \cdot I_{p,q}))&= (I_{p,q}\cdot g \cdot I_{p,q}\cdot I_{p,q}\cdot v_{\alpha}) \otimes (g\cdot (v_i\cdot v_j^T-v_i\cdot v_j^T) \cdot I_{p,q}\cdot I_{p,q} \cdot g^T \cdot I_{p,q}) \\
&=(I_{p,q}\cdot g \cdot  v_{\alpha}) \otimes (g\cdot (v_i\cdot v_j^T-v_i\cdot v_j^T)\cdot g^T \cdot I_{p,q}) \\
&=\lambda_{\alpha} \, \lambda_i \, \lambda_j \, ((I_{p,q}\cdot v_{\alpha}) \otimes ((v_i\cdot v_j^T-v_i\cdot v_j^T) \cdot I_{p,q}))
\end{aligned}
$$
Which is what we needed to show since it is already a basis of $\mathbf{R}^n \otimes \mathfrak{so}(p,q)$.
\end{proof}
We are now ready to start the reductions.

\begin{proposition}[First reduction of a $SO^+(p,q)$-equivariant $1$-form]
Given a $SO^+(p,q)$ principal bundle with an equivariant section bundle covering on $M_{p,q}$. Then a local description of a connection form $B$ is determined by $B(r\,e_1)$. Furthermore, $B(r\,e_1)$ is invariant by the $\rho$ action of $\Tilde{SO^+}(p-1,q):=\{Diag(1,P)|P\in SO^+(p-1,q)\}=\{g\in SO^+(p,q)|g\cdot e_1=e_1\} \simeq SO^+(p-1,q)$.
\end{proposition}
\begin{proof}
Using \ref{p,q1stred} we know that for $v\in M_{p,q}$
$$
B(v)=\rho(v)(B(||v||\,e_1))
$$
Which proves our first point. Then, since $r\,e_1$ is invariant for all $g\in\Tilde{SO}(p-1,q)$, so is $B(r\,e_1)$ under the $\rho$ action.
\end{proof}
Now we need to characterize $\bigcap_{g\in\Tilde{SO^+}(p-1,q)} ker(\rho(g)-Id)$.

\begin{proposition}
We have the following $1$-eigenspaces:
\begin{table}[h!]
\caption{$1$-eigenspaces of $SO^+(p,q)$'s standard elements}
\begin{center}
\begin{tabular}{c|c}
    $\rho(g_{i,j}^{\theta})$ & $\begin{aligned} &span\{(I_{p,q}\cdot e_k \otimes e^A_{l,m}\cdot I_{p,q})^{\{k,l,m\}\cap \{i,j\}= \emptyset }_{l<m} , \, (I_{p,q}\cdot e_i \otimes e^A_{i,k}\cdot I_{p,q} + I_{p,q}\cdot e_j \otimes e^A_{j,k}\cdot I_{p,q})_{k\notin \{i,j\}}, \\
& (I_{p,q}\cdot e_i \otimes e^A_{j,k} \cdot I_{p,q} -I_{p,q}\cdot  e_j \otimes e^A_{i,k}\cdot I_{p,q})_{k\notin \{i,j\}},\, (I_{p,q} \cdot e_k\otimes e^A_{i,j} \cdot I_{p,q} )_{k \notin \{i,j\}} \} 
\end{aligned}$ \\
     \hline
    $\rho(k_{i,j}^h)$ & $\begin{aligned} &span\{(I_{p,q}\cdot e_k \otimes e^A_{l,m}\cdot I_{p,q})^{\{k,l,m\}\cap \{i,j\}= \emptyset }_{l<m} , \, (I_{p,q}\cdot e_i \otimes e^A_{i,k}\cdot I_{p,q} - I_{p,q}\cdot e_j \otimes e^A_{j,k}\cdot I_{p,q})_{k\notin \{i,j\}}, \\
& (I_{p,q}\cdot e_i \otimes e^A_{j,k} \cdot I_{p,q} -I_{p,q}\cdot  e_j \otimes e^A_{i,k}\cdot I_{p,q})_{k\notin \{i,j\}},\, (I_{p,q} \cdot e_k\otimes e^A_{i,j} \cdot I_{p,q} )_{k \notin \{i,j\}} \} \end{aligned} $
\end{tabular}
\end{center}
\end{table}
\end{proposition}
\begin{proof}
Firstly, we can use what we did in the $SO(n)$ case for $g_{i,j}^{\theta}$.

Secondly, for $k^h_{i,j}$ we have the following:
\begin{table}[h!]
\caption{Eigenvalues and vectors of $SO^+(p,q)$'s standard elements' representations}
\begin{center}
        \begin{tabular}{|c|c|c||c|}
        \hline
        $\lambda_{\alpha}$ & $\lambda_l$ & $\lambda_m$ & \\
        \hline
        \hline
        $1$ & $1$ & $1$ & $I_{p,q}\cdot e_{\alpha} \otimes e_{l,m}\cdot I_{p,q}$ \\
        \hline
        $1$ & $e^{h}$ & $e^{-h}$ & $I_{p,q}\cdot e_{\alpha}\otimes(2\,e^A_{i,j}\cdot I_{p,q})$\\
        \hline
        $e^{h}$ & $1$ & $e^{-h}$ & $I_{p,q}\cdot (e_i + e_j) \otimes (e_{\alpha,i}^A - e_{\alpha,j}^A)\cdot I_{p,q} $ \\
        \hline
        $e^{h}$ & $1$ & $e^{-h}$ & $I_{p,q}\cdot (e_i - e_j) \otimes (e_{\alpha,i}^A + e_{\alpha,j}^A)\cdot I_{p,q} $\\
        \hline
        \end{tabular}
    \end{center}
\end{table}
    We then add and substract the last two to get the wanted $1$-eigenspace.
\end{proof}
\begin{remark}\label{Ipqdentities}
We have the following identities for $i,j \leq p$ and $k,l \geq p+1$
\begin{enumerate}
    \item $I_{p,q}\cdot e_i=e_i$
    \item $I_{p,q} \cdot e_k= -e_k$
    \item $e^A_{i,j} \cdot I_{p,q}=e^A_{i,j}$
    \item $e^A_{k,l} \cdot I_{p,q}= -e^A_{k,l}$
    \item $e^A_{i,k} \cdot I_{p,q}= -e^S_{i,k}$
\end{enumerate}
\end{remark}
All we need to do now is to intersect them. Fortunately for us we already did half the work in the $SO(n)$ case.

If we write $G_{i,j}$ the $1$-eigenspace of $\rho(g_{i,j}^{\theta})$ and $K_{i,j}$ the $1$-eigenspace of $\rho(k_{i,j}^h)$.

Regarding $\rho(k_{i,j}^h)$ we want to find $\bigcap_{j=p+1}^{p+q} \, \bigcap_{i=2}^p K_{i,j}$ so we will first look at, for $p<j$ 
$$
\bigcap_{i=2}^p K_{i,j}
$$
To do so we will need:
\begin{lemma}[First intersection]\label{Kintersect}
For $j>p$ and $2\leq i\leq p-1$ we have
$$
\begin{aligned}
K_{i,j}\cap K_{i+1,j}= span\{&(I_{p,q}\cdot e_k \otimes e^A_{l,m}\cdot I_{p,q})^{\{k,l,m\}\cap \{i,i+1,j\}= \emptyset}_{l<m} \\
& (e_i \otimes e^A_{i,k}\cdot I_{p,q} + e_{i+1} \otimes e^A_{i+1,k}\cdot I_{p,q} - e_j\otimes e^A_{j,k}\cdot I_{p,q})_{k\notin \{i,i+1,j\}} \\
& -e_{i+1}\otimes e^S_{i,j} +e_j\otimes e^A_{i,i+1} + e_i \otimes e^S_{i+1,j} \}
\end{aligned}
$$
\end{lemma}
\begin{proof}
As we did in the $SO(n)$ case, we just take $x \in K_{i,j}\cap K_{i+1,j}$ and project the two different decomposition to get the simplifying relations. To go further we then apply the identities of remark \ref{Ipqdentities}.
\end{proof}

\begin{proposition}[The different cases of $SO^+(p,q)$]\label{caseSOpq}
Regarding the $\rho(g_{i,j}^{\theta})$ we have:
$$
\bigcap_{2\leq i<j \leq p} G_{i,j}=\bigcap_{j=3}^p G_{2,j}=\begin{cases}
\mathbf{R}^n \otimes \mathfrak{so}(p,q) & p\leq 2 \\
 G_{2,3} & p=3
\\
\begin{aligned}
span \{& (I_{p,q}\cdot e_k \otimes e^A_{l,m} \cdot I_{p,q} )_{l<m}^{\{k,l,m\}\cap \{2,3,4\} =\emptyset},\\
& (e_2 \otimes e^A_{2,k} \cdot I_{p,q} + e_3 \otimes e^A_{3,k} \cdot I_{p,q} \\
& + e_{4} \otimes e^A_{4,k}\cdot I_{p,q})_{k \notin \{2,3,4\}} , \\
& \, e_2 \otimes e^A_{3,4} - e_3 \otimes e^A_{2,4} + e_{4} \otimes e^A_{2,3} \}
\end{aligned} & p=4 \\
\begin{aligned}  span\{&(I_{p,q}\cdot e_k\otimes e^A_{l,m} \cdot I_{p,q})^{\{k,l,m\}\cap \{2,...,p\}=\emptyset}_{l<m},\\
    & (\sum_{l=2}^{p} I_{p,q}\cdot  e_l \otimes e^A_{l,k} \cdot I_{p,q})_{k\notin \{2,...,p\}} \} \end{aligned} & p\geq 5
\end{cases}
$$
$$
\bigcap_{p+1\leq i<j \leq p+q} G_{i,j}=\bigcap_{j=p+2}^{p+q} G_{2,j}=\begin{cases}
\mathbf{R}^n \otimes \mathfrak{so}(p,q) & q=1 \\
 G_{p+1,p+2} & q=2
\\
 \begin{aligned}
span \{& (e_k \otimes e^A_{l,m})_{l<m \leq p}^{k\leq p},\\
& (e_{p+1} \otimes e^S_{p+1,k} + e_{p+2} \otimes e^S_{p+2,k} \\
&+ e_{p+3} \otimes e^S_{p+3,k})_{k \leq p} , \\
& \, e_{p+1} \otimes e^A_{p+2,p+3} - e_{p+2} \otimes e^A_{p+1,p+3} \\
& + e_{p+3} \otimes e^A_{p+1,p+2} \}
\end{aligned} & q=3 \\
\begin{aligned}  span\{&(e_k\otimes e^A_{l,m})^{k\leq p}_{l<m\leq p}\\
    & (\sum_{l=p+1}^{p+q} e_l \otimes e^S_{l,k})_{k\leq p} \} \end{aligned} & q\geq 4
\end{cases}
$$
Regarding $\rho(k_{i,j}^h)$ we want to find $\bigcap_{j=p+1}^{p+q} \, \bigcap_{i=2}^p K_{i,j}$ so we will first look at, for $p<j$:
$$
\bigcap_{i=2}^p K_{i,j}= \begin{cases}
\mathbf{R}^n \otimes \mathfrak{so}(p,q) & p=1 \\
K_{2,j} & p=2 \\
\begin{aligned} 
span\{&(I_{p,q}\cdot e_k \otimes e^A_{l,m}\cdot I_{p,q})^{\{k,l,m\}\cap \{2,3,j\}= \emptyset}_{l<m} \\
& (e_2 \otimes e^A_{2,k}\cdot I_{p,q} + e_{3} \otimes e^A_{3,k}\cdot I_{p,q} - e_j\otimes e^A_{j,k}\cdot I_{p,q})_{k\notin \{2,3,j\}} \\
& -e_{3}\otimes e^S_{2,j} +e_j\otimes e^A_{2,3} + e_2 \otimes e^S_{3,j} \}
\end{aligned} & p=3 \\
\begin{aligned} 
span\{ &(I_{p,q}\cdot e_k\otimes e_{l,m}^A \cdot I_{p,q})_{l<m}^{\{k,l,m\}\cap (\{2,...,p\}\cup \{j\})=\emptyset} \\
& (\sum_{i=2}^p e_i\otimes e^A_{i,k}\cdot I_{p,q} - e_j\otimes e^A_{j,k}\cdot I_{p,q})_{k\notin \{2,...,p\}\cup \{j\}} \}
\end{aligned} & p\geq 4
\end{cases} 
$$

\end{proposition}

\begin{proof}
For the intersections of $G_{i,j}$ this is a direct use of what was done in the proof of theorem \ref{intersecteigen}.

As for the intersection of $K_{i,j}$ we use lemma \ref{Kintersect} for the first two cases and then:
\begin{itemize}
    \item[$p=3$] Notice that similarly to theorem \ref{intersecteigen} 
    $$
    K_{i,j}\cap K_{i+1,j} \cap K_{i+2,j}=(K_{i,j}\cap K_{i+1,j})\cap ( K_{i+1,j} \cap K_{i+2,j})
    $$
    We then use lemma \ref{Kintersect} to write any element on the two basis before projecting to get the needed relations.
    \item[$p\geq 4$] We prove it by induction. The proof is exactly the same as the one in theorem \ref{intersecteigen}.
\end{itemize}
\end{proof}
We are now sufficiently equipped to find the global intersection.

\begin{theorem}[The different cases of $SO^+(p,q)$]
Here are the values of $\bigcap_{g\in \Tilde{SO^+}(p-1,q)} ker(\rho(g)-Id)$ depending on $p$ and $q$.

\begin{itemize}
\item If
    $p=1$ and $q=1$ then it is equal to $\mathbf{R}\otimes \mathfrak{so}(1,1)$ 
    \item
     If $p=1$ and $q=2$ then it is equal to $span\{(e_2\otimes e^S_{2,1}-e_3\otimes e^S_{3,1}),\,(e_2\otimes e^S_{3,1}-e_3\otimes e^S_{2,1}),\, (e_1\otimes e^A_{2,3})\}$ 
     \item
     If $p=1$ and $q=3$ then it is equal to $span\{(e_2\otimes e^S_{2,1} + e_3\otimes e^S_{3,1} - e_4\otimes e_{4,1}^S),\, (e_2 \otimes e^A_{3,4} - e_3 \otimes e^A_{2,4} + e_4 \otimes e_{2,3}^A)\}$ 
     \item
     if $p=2$ and $q=2$ then it is equal to $span\{(e_3\otimes e^S_{3,1}+ e_4\otimes e^S_{4,1} - e_2 \otimes e^A_{2,1}),\,(-e_4\otimes e^S_{3,2} + e_2\otimes e_{3,4}^A + e_3\otimes e^S_{4,2})\}$
     \item
     Else it is equal to $span\{\sum_{i=2}^p e_i \otimes e^A_{i,1} - \sum_{j=p+1}^{p+q} e_j \otimes e_{j,1}^S\}$
\end{itemize}
\end{theorem}
\begin{proof}
To get the result we need to consider the following cases individually with the results in proposition \ref{caseSOpq}:
\begin{enumerate}
    \item[$p=1$,\,] 
    \begin{enumerate}
        \item[$q=1$] It is immediate from what we have shown.
        \item[$q=2$] This is just a rewriting of $G_{2,3}$ in the case where $p=1$.
        \item[$q=3$] It follows directly from the corresponding case in \ref{caseSOpq}.
    \end{enumerate}
    \item[$p=2$,\,]
    \begin{enumerate}
        \item[$q=2$] We want to compute $G_{3,4}\cap K_{2,3} \cap K_{2,4}$. To do so we use lemma \ref{Kintersect} on $K_{3,2}\cap K_{4,2}$ and then intersect with $G_{3,4}$ as we have done before.
        \item[$q=3$] We want to compute 
        $$\begin{aligned}
span \{& (e_k \otimes e^A_{1,2})^{k\leq 2},\, (e_{3} \otimes e^S_{3,k} + e_{4} \otimes e^S_{4,k} \\
&+ e_{5} \otimes e^S_{5,k})_{k \leq p} , \, e_{3} \otimes e^A_{4,5} - e_{4} \otimes e^A_{3,5} \\
& + e_{5} \otimes e^A_{3,4} \} \cap K_{2,3}\cap K_{2,4} \cap K_{2,5}
\end{aligned}
$$
Firstly, we can use the computations we have already done to see that $K_{2,3}\cap K_{2,4} \cap K_{2,5}=span\{\sum_{i=3}^5 e_i \otimes e^S_{i,1} + e_2 \otimes e_{2,1}^A\}$. We then check that the vector is indeed in the first vector space.
\item[$q\geq 4$] The results comes the same way as above.
    \end{enumerate}
    \item[$p=3$,\,]
    \begin{enumerate}
        \item[$q=3$] We first need to compute
        $$
        \begin{aligned}
        \bigcap_{j=4}^6 \bigcap_{i=2}^3 K_{i,j}= span\{&(I_{p,q}\cdot e_k \otimes e^A_{l,m}\cdot I_{p,q})^{\{k,l,m\}\cap \{2,3,j\}= \emptyset}_{l<m} \\
& (e_2 \otimes e^A_{2,k}\cdot I_{p,q} + e_{3} \otimes e^A_{3,k}\cdot I_{p,q} - e_j\otimes e^A_{j,k}\cdot I_{p,q})_{k\notin \{2,3,j\}} \\
& e_{3}\otimes e^S_{2,j} +e_j\otimes e^A_{2,3} - e_2 \otimes e^S_{3,j} \}
\end{aligned}
        $$
        This gives directly $span\{\sum_{i=2}^3 e_i \otimes e^A_{i,1} + \sum_{j=4}^6 e_j\otimes e_{j,1}^S$. We then check that it is in all the other $1$-eigenspaces.
        \item[$q\geq 4$] It follows the same way.
    \end{enumerate}
    \item[$p\geq 4$,\,] We directly use \ref{caseSOpq} to get the result.
\end{enumerate}
We then check individually that the eigenvectors are indeed invariant for all $g\in \Tilde{SO}(p,q)$.
\end{proof}

Introduce $\epsilon (i)=\begin{cases}
1 & i\leq p \\
-1 & i> p
\end{cases}$

We can now find the general form of our $SO(p,q)$-equivariant connection form.

\begin{theorem}[$SO^+(p,q)$-equivariant $1$-form]\label{SO(p,q)-reduction}

Given a $SO^+(p,q)$ $1$-form which is $SO^+(p,q)$-equivariant then it is of the form (for $r=||x||$):

\begin{itemize}
       \item If $p+q=3$  we have $B(x)=f(r)\,\epsilon(\mu)\,\epsilon(i)  \, x_{\mu} \,\sum_{k=1}^3 \epsilon_{i,j,k}\,x_k + g(r)\,\epsilon(j) \, (\delta_{\mu}^j\,x_i-\delta_{\mu}^i\,x_j) + h(r)\,\epsilon(\mu)\, \epsilon(j)\,\epsilon_{\mu,i,j}$
    \item If
        $p+q=4$ then $B(x)=f(r)\,\epsilon(i) \, \sum_{k=1}^4 \epsilon_{k,\mu,i,j} \, x_k + g(r)\,\epsilon(j) \, (\delta_{\mu}^j\,x_i-\delta_{\mu}^i\,x_j)$ 
    \item Else
        $p+q\geq 5$ and $B(x)=g(r)\,\epsilon(j) \, (\delta_{\mu}^j\,x_i-\delta_{\mu}^i\,x_j)$
    \end{itemize}
\end{theorem}
\begin{proof}
The proof is exactly like the two previous ones:
Let us start the computations: take $g_x$ given by \ref{p,q1stred}: we have
\begin{itemize}
\item In the case $(p,q)=(1,2)$ we can use Gram-Schmidt to get, up to a sign in the last column: 
$$
g_x=\begin{pmatrix}
 \frac{x_1}{||x||} & \frac{x_1\,x_2}{||x||\,\sqrt{x_1^2-x_3^2}} & \frac{x_3}{\sqrt{x_1^2-x_3^2}} \\
 \frac{x_2}{||x||} & \frac{x_1^2-x_3^2}{||x||\,\sqrt{x_1^2-x_3^2}} & 0\\
 \frac{x_2}{||x||} & \frac{x_2\,x_3}{||x||\,\sqrt{x_1^2-x_3^2}} & \frac{x_1}{\sqrt{x_1^2-x_3^2}}
\end{pmatrix}
$$
$$
\begin{aligned}
\rho(g_x)(e_1\otimes e_{2,3}^A) &=\epsilon(\mu)\,\epsilon(i)\sum_{k,l,\nu=1}^3 [g_x]_{\mu,\nu} \, \delta_{\nu,1} \, [g_x]_{i,k} \, [g_x\cdot I_{p,q}]_{j,l} (\delta_k^2 \, \delta_l^3 - \delta_k^3 \, \delta_l^2) \\
&= \epsilon(\mu)\,\epsilon(i)  \,\frac{x_{\mu}}{||x||} \,([g_x]_{i,3}\,[g_x]_{j,2} - [g_x]_{i,2}\,[g_x]_{j,3}) \\
&= \epsilon(\mu)\,\epsilon(i)  \, \frac{x_{\mu}}{||x||} \,\sum_{k=1}^3 \epsilon_{i,j,k}\,\frac{x_k}{||x||}
\end{aligned}
$$
then,
$$
\begin{aligned}
\rho(g_x)( -e_2\otimes e^S_{3,1} + e_3 \otimes e^S_{2,1}+e_1\otimes e_{2,3}^A) &= \epsilon(\mu)\, \epsilon(i)\, \sum_{k,l,\nu=1}^3 [g_x]_{\mu,\nu}\, [g_x]_{i,k}\, [g_x]_{j,l}\, \epsilon_{\nu,k,l}\\
&=\epsilon(\mu)\, \epsilon(i)\,\epsilon_{\mu,i,j}\, \sum_{k,l,\nu=1}^3 [g_x]_{1,\nu}\, [g_x]_{2,k}\, [g_x]_{3,l}\, \epsilon_{\nu,k,l} \\
&=\epsilon(\mu)\, \epsilon(i)\,\epsilon_{\mu,i,j}
\end{aligned}
$$
\item When $(p,q)=(1,3)$ we can take:
$$
g_x=\begin{pmatrix}
 \frac{x_1}{||x||} & \frac{x_1\,x_2}{||x||\,\sqrt{x_1^2-x_3^2-x_4^2}} &  \frac{x_1\,x_3}{\sqrt{x_1^2-x_3^2-x_4^2}\,\sqrt{x_1^2-x_4^2}} & \frac{x_4}{\sqrt{x_1^2-x_4^2}} \\
 \frac{x_2}{||x||} & \frac{x_1^2-x_3^2-x_4^2}{||x||\,\sqrt{x_1^2-x_3^2-x_4^2}} &  0 & 0 \\
 \frac{x_3}{||x||} & \frac{x_2\,x_3}{||x||\,\sqrt{x_1^2-x_3^2-x_4^2}} &  \frac{x_1^2-x_4^2}{\sqrt{x_1^2-x_3^2-x_4^2}\,\sqrt{x_1^2-x_4^2}} & 0 \\
 \frac{x_4}{||x||} & \frac{x_2\,x_4}{||x||\,\sqrt{x_1^2-x_3^2-x_4^2}} &  \frac{x_3\,x_4}{\sqrt{x_1^2-x_3^2-x_4^2}\,\sqrt{x_1^2-x_4^2}} & \frac{x_1}{\sqrt{x_1^2-x_4^2}} \\
\end{pmatrix}
$$
Then we have:
$$
\begin{aligned}
\rho(g_x)(e_2 \otimes e^A_{3,4} - e_3 \otimes e^A_{2,4} + e_4 \otimes e_{2,3}^A) &= \epsilon(\mu)\,\epsilon(j)\, \sum_{\nu,k,l=2}^4 \epsilon_{1,\nu,k,l}\,[g_x\cdot I_{p,q}]_{\mu,\nu}\,[g_x]_{i,k}\,[g_x\cdot I_{p,q}]_{j,l} \\
&= \epsilon(\mu)\,\epsilon(j)\, \sum_{\nu,k,l=2}^4 \epsilon_{1,\nu,k,l}\,[g_x]_{\mu,\nu}\,[g_x]_{i,k}\,[g_x]_{j,l}\\
&= \epsilon(\mu)\, \epsilon(j)\, \sum_{k=1}^4 \epsilon(k)\, \epsilon_{k,\mu,i,j} \, \frac{x_k}{||x||}\\
&=-\epsilon(i) \,\sum_{k=1}^4 \epsilon_{k,\mu,i,j} \, \frac{x_k}{||x||}
\end{aligned}
$$
\item When $(p,q)=(2,2)$ we can take
$$
g_x=\begin{pmatrix}
 \frac{x_1}{||x||} & \frac{x_1\,x_2}{||x||\,\sqrt{x_1^2-x_3^2-x_4^2}} &  \frac{x_1\,x_3}{\sqrt{x_1^2-x_3^2-x_4^2}\,\sqrt{x_1^2-x_4^2}} & \frac{x_4}{\sqrt{x_1^2-x_4^2}} \\
 \frac{x_2}{||x||} & \frac{-x_1^2+x_3^2+x_4^2}{||x||\,\sqrt{x_1^2-x_3^2-x_4^2}} &  0 & 0 \\
 \frac{x_3}{||x||} & \frac{x_2\,x_3}{||x||\,\sqrt{x_1^2-x_3^2-x_4^2}} &  \frac{x_1^2-x_4^2}{\sqrt{x_1^2-x_3^2-x_4^2}\,\sqrt{x_1^2-x_4^2}} & 0 \\
 \frac{x_4}{||x||} & \frac{x_2\,x_4}{||x||\,\sqrt{x_1^2-x_3^2-x_4^2}} &  \frac{x_3\,x_4}{\sqrt{x_1^2-x_3^2-x_4^2}\,\sqrt{x_1^2-x_4^2}} & \frac{x_1}{\sqrt{x_1^2-x_4^2}} \\
\end{pmatrix}
$$
Then we have:
$$
\begin{aligned}
\rho(g_x)(e_4\otimes e^S_{3,2} - e_2\otimes e_{3,4}^A - e_3\otimes e^S_{4,2})&=\epsilon(\mu)\,\epsilon(j)\, \sum_{\nu,k,l=2}^4 -\epsilon(l)\, \epsilon_{1,\nu,k,l}\,[g_x]_{\mu,\nu}\,[g_x]_{i,k}\,[g_x]_{j,l} \\
&=\epsilon(\mu)\, \epsilon(j)\, \sum_{k=1}^4 \epsilon(k)\, \epsilon_{k,\mu,i,j} \, \frac{x_k}{||x||}\\
&=\epsilon(i) \, \sum_{k=1}^4 \epsilon_{k,\mu,i,j} \, \frac{x_k}{||x||}
\end{aligned}
$$
 \item In all cases we have got
$$
\begin{aligned}
\rho(g_x)(\sum_{k=1}^p e_k \otimes e^A_{k,1} - \sum_{k=p+1}^{p+q} e_k \otimes e_{k,1}^S) &=\epsilon(\mu)\,\epsilon(j)\, \sum_{k,l,m,\nu=1}^n [g_x\cdot I_{p,q}]_{\mu,\nu} \, [g_x]_{i,l} \, [g_x \cdot I_{p,q}]_{j,m} \\
&\delta_k^{\nu} \, (\delta^k_l\,\delta^1_m - \epsilon(k) \, \delta_l^1 \, \delta_k^m) \\
&= \epsilon(\mu) \, \epsilon(j) \,(\sum_{\nu=1}^n [g_x\cdot I_{p,q}]_{\mu,\nu}\, [g_x]_{i,\nu} \, [g_x\cdot I_{p,q}]_{j,1}\\
&- \sum_{\nu=1}^n [g_x\cdot I_{p,q}]_{\mu,\nu}\, \epsilon(\nu) \, [g_x]_{i,1} \, [g_x\cdot I_{p,q}]_{j,\nu} )\\
&=\epsilon(\mu) \, \epsilon(j) \,([g_x\cdot I_{p,q} \cdot g_x^T]_{\mu,j}\,\frac{x_i}{||x||}\\
&-[g_x\cdot I_{p,q} \cdot g_x^T]_{\mu,i}\,\frac{x_j}{||x||}) \\
&=\epsilon(j) \, (\delta_{\mu}^j\,\frac{x_i}{||x||}-\delta_{\mu}^i\,\frac{x_j}{||x||})
\end{aligned}
$$
\end{itemize}

\end{proof}
\begin{remark}
One can notice that a general form for such $g_x$ is given by
$$
[g_x]_{i,j}=\begin{cases}
\frac{x_i}{||x||} & j=1 \\
\frac{x_1\,x_j}{\sqrt{x_1^2 + \sum_{k=j}^n \epsilon(k)\,x_k^2}\,\sqrt{x_1^2 + \sum_{k=j+1}^n \epsilon(k)\,x_k^2}} & i=1,\,j>1 \\
\frac{-\epsilon(i)\,x_i\,x_j}{\sqrt{x_1^2 + \sum_{k=j}^n \epsilon(k)\,x_k^2}\,\sqrt{x_1^2 + \sum_{k=j+1}^n \epsilon(k)\,x_k^2}} & 1<i<j \\
0 & 1<j<i\\
\frac{-\epsilon(i)\,\sqrt{x_1^2 +\sum_{k=j+1}^n \epsilon(k)\,x_k^2}}{\sqrt{x_1^2 + \sum_{k=j}^n \epsilon(k)\,x_k^2}} & 1<i=j
\end{cases}
$$
\end{remark}

\subsection{$SO(n)$ simplifications to an $SU(n)$-connection}\label{SU(n)}
In this subsection we will consider a linear $1$-form $B_{su}\in\Omega^1(M,\mathfrak{su}(n))$ and use the injection $SO(n) \subset SU(n)$ to define $SO(n)$-equivariance.

Firstly, notice that we can always write $B_{su}$ as:
$$
B_{su}=B_{so}+i\,S
$$
Where $B_{so}$ lives in $\Omega^1(U,\mathfrak{so}(n))$ and each component of $S$ is valued in the real symmetric matrices with its trace equal to $0$.

\begin{proposition}
With the notations introduced, $B_{su}$ is $SO(n)$ equivariant if and only if both $B_{so}$ and $S$ are $SO(n)$ equivariant.
\end{proposition}
\begin{proof}
This comes directly from the fact that $SO(n)$ acts independently on real and imaginary parts.
\end{proof}

Since we already know what $B_{so}$ looks like, we need to characterize $S$.

\begin{proposition}
For $n\geq 4$, $S$ is of the form:
$$
S_i(x)=g_1(r)\,x_i\,\left ( x\cdot x^T-\frac{r^2}{n}\,Id \right ) +g_2(r)\,\left ( e_i\cdot x^T + x\cdot e_i^T - \frac{2\,x_i}{n}\,Id \right )
$$
\end{proposition}
\begin{proof}
Let us now look at $S_i(r\,e_1)$. With the same reasoning as before we need to find the intersection of the $1$-eigenspaces of $\rho(g_{i,j}^{\theta})$ with the traceless symmetric matrices. An eigenvector is of the form $v_k\otimes (v_i\cdot v_j^T + v_j\cdot v_i^T)$. Doing so and imposing a traceless condition we find that:
$$
S_i(t,r\,e_1)\in span\{e_1\otimes \left (e_{1,1}-\frac{1}{n}\, Id\right ), \sum_{i=1}^n e_i\otimes e^S_{i,1}-\frac{2}{n}\,e_1\otimes Id\} 
$$
We then apply $g_x$ such that $g_x\cdot e_1=\frac{x}{|x|}$ to get the wanted result.
\end{proof}

\begin{remark}
We only did the case $n\geq 4$ because the case $n=3$ gives us two extra terms that are hard to work with.
\end{remark}

\begin{corollary}[General form for an $SO(n)$ equivariant $SU(n)$ connection form]
If we impose $n\geq 5$ or certain functions to be null then the general form for $SO(n)$ equivariant $SU(n)$ connection form is:
$$
\begin{aligned}
(B_{su})_i(x)=&f(r)\,\left ( x\cdot e_i^T-e_i\cdot x^T \right )  \\
+&i\, \left (g_1(r)\,x_i\,\left ( x\cdot x^T-\frac{r^2}{n}\,Id \right ) +g_2(r)\,\left ( e_i\cdot x^T + x\cdot e_i^T - \frac{2\,x_i}{n}\,Id \right ) \right ) \\
\end{aligned}
$$
\end{corollary}
\newpage
 
\section{Simplifications to the local connection form}
To give an abstract justification for the notion of equivariance, we'll now place ourselves on a small enough $U_i$ with a given $s_i$ such that we're on a chart of the form $(U_i,\phi_i =(x_1,...,x_n))$. To simplify, we'll denote $A^{s_i}$ as $B\in \Omega^1 (U_i,\mathfrak{g})$. Further conditions will add on as we progress throughout the section.
\subsection{General simplifications}
\begin{proposition}[First simplification]
If we write $B_{\mu}=B(\partial _{\mu})$ and we have $f_1,...,f_r$ a basis of $\mathfrak{g}$ then we can write
$$
\begin{aligned}
B &=\sum_{a=1}^r \sum_{\mu =1}^n B_{\mu}^a \, f_a \, dx^{\mu} \\
&= \sum_{\mu=1}^n B_{\mu}\, dx^{\mu} \\
&= \sum_{a=1}^r B^a \, f_a
\end{aligned}
$$
\end{proposition}

We will now suppose further structure with a right transformation group action of $G$, $r_g$ on $M$. We will take inspiration from \cite{Gastel2011YangMillsCO} to have the same kind of effective simplification.

\begin{definition}[Adapted Equivariant section] \label{eqsec}
Given a connection $1$-form $A$ and $H$ a subgroup of $G$. We will say that a section $s:U\longrightarrow P$ is an \textbf{adapted $H$-equivariant section} of $A$ if for all $g\in H$
$$
A^{R_g\circ s}=A^{s \circ r_g}
$$
We will often shorten this to say that $s$ or $A$ is $H$-equivariant.
\end{definition}
\begin{remark}
This is the case when $R_g\circ s= s \circ r_g$ but this would imply $r_g=Id$ by applying $\pi$ on both sides.
\end{remark}
\begin{remark}
This represents a connection which preserves the additional structure on the base manifold, similarly to Levi-Civita connections who preserve the metric structure on the base Manifold.
\end{remark}

We then have the following simplification

\begin{proposition}[Equivariant simplification]\label{Equisimp}
Given an equivariant section bundle, the local $1$-forms $B=A^s$ satisfy:
$$
Ad(g^{-1})(B(x))=B(x\cdot g)\circ T_x r_g
$$
\end{proposition}
\begin{proof}
If we just apply $s^*$ to the original equation we get for all $x\in M$:
$$
\begin{aligned}
Ad(g^{-1}) (B(x))&=(R_g ^* A)^s(x)=A_{s(x)\cdot g}\circ T_{s(x)}R_g\circ T_x s \\
&= A_{R_g \circ s(x)}\circ T_x(R_g\circ s)= A^{R_g\circ s}=A^{s\circ r_g} \\
&= A_{s(x\cdot g)}\circ T_x(s\circ r_g)=A_{s(x\cdot g)}\circ T_{x\cdot g}s \circ T_x r_g \\
&= A^s(x\cdot g) \circ T_x r_g= B(x\cdot g) \circ T_x r_g
\end{aligned}
$$
\end{proof}

\subsection{Matrix group simplifications}
We now place ourselves in the case of $G\subset GL(m,\mathbf{K})$. We then have
\begin{proposition}[Matrix simplifications]\label{matrix}
By linearity we have for all $X\in \mathfrak{g}$
\begin{enumerate}
    \item $TL_g X=g\cdot X$ when $L_g$ is constant
    \item $Ad(g)(X)=g\cdot X \cdot g^{-1}$
\end{enumerate}
This allows us to rewrite the equivariance condition \ref{connection} of \ref{1form} as, for all $g\in G$
$$
\begin{aligned}
\sum_{\mu=1}^n g^{-1}\cdot B_{\mu}(x) \cdot g \, dx^{\mu}&= A(s_i (x) \cdot g) \circ T_{s_i (x) \cdot g} R_g \circ T_x s_i \\
&=A(s_i (x) \cdot g) \circ T_x (R_g \circ s_i)
\end{aligned}
$$
\end{proposition}

But now we want to get rid of $A$ and come back to $B$. To do so \textbf{we'll now suppose $m=n$}. This means that we can define a local right action on $M$ by 
$$
\forall g\in G,\, x=(x_1,...,x_n) \in U_i, \,x\cdot g= g^{-1} \cdot x
$$
Where $g^{-1} \cdot x$ is the comparable to the matrix multiplication in coordinates.

We'll now suppose that the section $s_i$ we consider is $H$-equivariant.

\begin{proposition}[Equivariant section simplification]\label{eqsecsimp}
If we suppose the section $s_i$ to be $H$-equivariant then the condition \ref{connection} gives for all $g\in H$:
$$
B(L_g(x))=Ad(g)(B(x))\circ (T_x L_{g})^{-1}
$$
\end{proposition}
\begin{proof}
Using the simplifications in \ref{matrix} and the definition \ref{Equisimp} we have
$$
Ad(g^{-1})(B(x))=B(x\cdot g)\circ T_x r_g=B(x\cdot g)\circ T_x L_{g^{-1}}
$$
We then switch $g$ with $g^{-1}$ and rearrange to get the wanted form.
\end{proof}

\begin{corollary}[Equivariant section simplification with constant left action]If we suppose the section $s_i$ to be $H$-equivariant and the group actions to be constant, then condition \ref{connection} becomes $B$ is \textbf{algebraically $H$-equivariant}.
\end{corollary}

\begin{example}
To better interpret this, let us place ourselves in the context of an $SO^+(1,2)$ bundle such that locally $M$ looks like time-like coordinate space. We will consider a path $\gamma$ in $M$.

Then notice that with the correspondance seen in previous example we can define a local covariant derivative from $B$ which leads to the same corresponding parallel transport. Furthermore, given a local chart $\phi_{U}:U\rightarrow \mathbf{R}^{2,1}$ we get the following relation if $\nabla$ is $g$-equivariant:
$$
\nabla_{X\cdot g}^{Ad(g)(B)} ((\phi_{U}\circ \gamma(s))\cdot g)=\nabla_X^B (\phi_{U}\circ \gamma(s))
$$
Which can be seen through the following:
\begin{figure}[h!]
    \centering
    \caption{$g$-equivariant $SO^+(1,2)$ connection (in color)}
    \includegraphics[width=\textwidth]{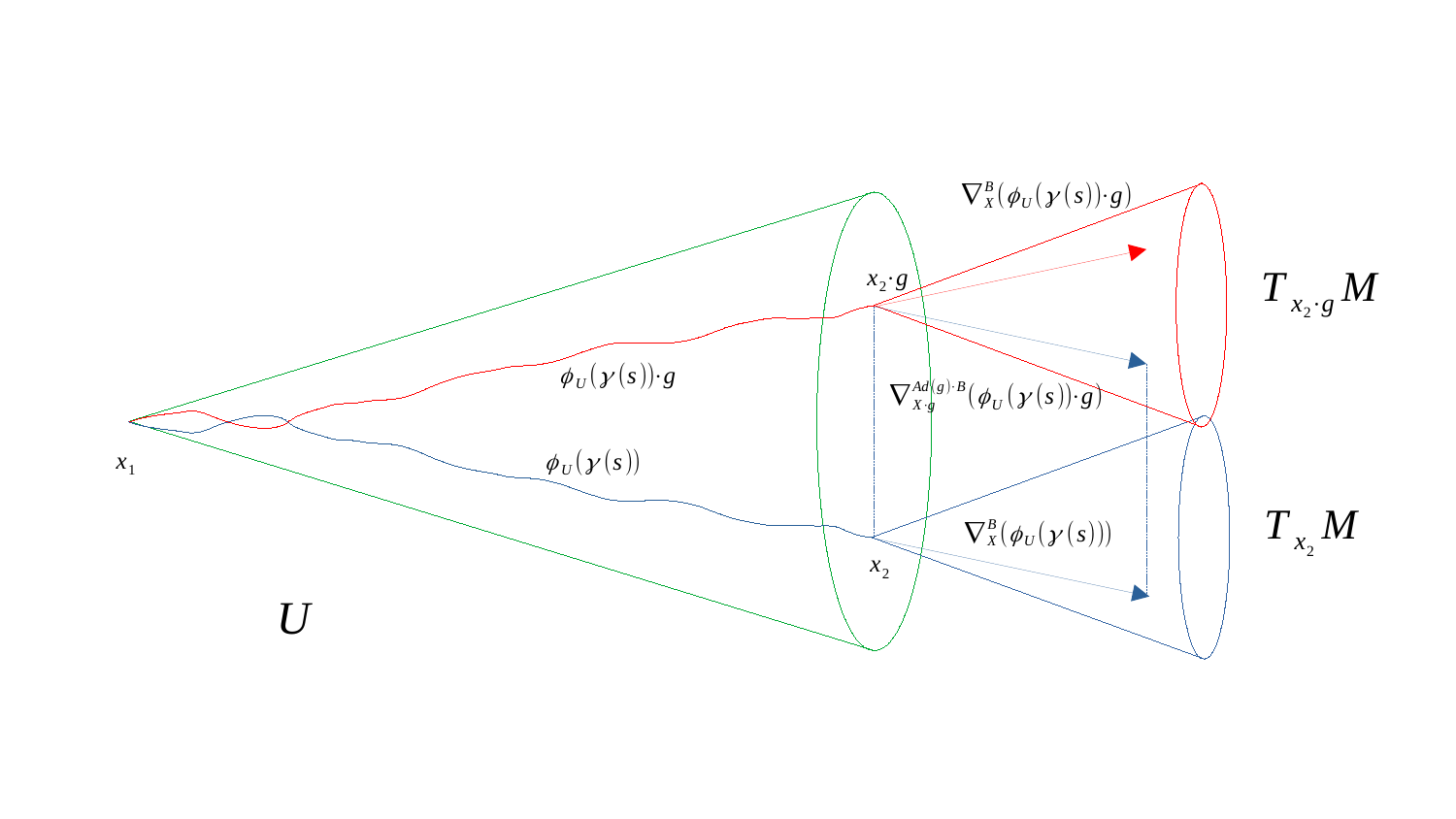}
\end{figure}
\end{example}

\subsection{$SO^+(p,q)$-equivariant $Spin^+(p,q)$ connections}\label{electroweak}
We now want to lift what we have found in $SO^+(p,q)$ to the spin group $Spin^+(p,q)$. We will briefly introduce the few properties we need on the spin group but to find definitions and general properties we recommend \cite{MathGauge}, \cite{MatrixGroup} and \cite{Spinpq}.

\begin{definition}[Clifford Algebra]
\textbf{The Clifford Algebra} $Cl(p,q)$ is the free $\mathbf{R}$-algebra factored by the relation:
$$
u \times v + v \times u= 2\,u^T\cdot I_{p,q} \cdot v \, 1
$$
\end{definition}
\begin{remark}
The basis $e_i$ of $\mathbf{R}^{p,q}$ are generators of $Cl(p,q)$ with the property $e_i\times e_j=-e_j\times e_i$.
\end{remark}

We will not define the spin group properly but give characteristics we will use:

\begin{proposition}[Spin group]
\textbf{The spin group} $Spin^+(p,q)$ is a smooth double cover of $SO^+(p,q)$ with $\lambda:Spin^+(p,q) \longrightarrow SO^+(p,q)$ and its differential $\lambda_d$. We have:
\begin{center}
\begin{tikzcd}
1 \arrow[r] & \mathbf{Z}_2 \arrow[r] & Spin^+(p,q) \arrow[r, "\lambda"] & SO(n) \arrow[r] & 1
\end{tikzcd}
\end{center}
Meaning that:
\begin{itemize}
    \item $ker(\lambda)=\{\pm 1\}$
    \item Its Lie algebra $\mathfrak{spin(p,q)}=span\{e_i\times e_j \in Cl(p,q)| 1\leq i <j \leq p+q\} \simeq \mathfrak{so}(p,q)$
    \item $\lambda_d(z)\,x=z\cdot x - x\cdot z$
    \item $\lambda_d(e_i\times e_j)=2\,e^A_{i,j}$
    \item $\lambda_d$ is a Lie algebra isomorphism with inverse: 
    
    $\lambda_d^{-1}(X)=\frac{1}{2}\,\sum_{k<l} ((X\cdot e_k)^T\cdot I_{p,q} \cdot e_l)\,\epsilon(k)\,\epsilon(l)\,e_k\times e_l$ 
\end{itemize}
\end{proposition}

\begin{definition}[Spinors]
Given a $(P,\pi,M)$ a $Spin^+(p,q)$ principal fiber bundle. We call \textbf{spinor} its sections.
\end{definition}
\begin{remark}
Sometimes the term is also used for the associated sections in the induced vector bundle.
\end{remark}
Now, we want to be able to derive said spinors when we already know some symmetries on local $SO^+(p,q)$ -connections.

\begin{definition}[The Induced $Spin^+(p,q)$ connection]
Given $B_{SO}$ the local description of a $SO^+(p,q)$ connection. We call \textbf{induced $Spin^+(p,q)$ connection} the connection induced by the local description $B_{spin}=\lambda_d^{-1}\circ (B_{SO}) \in \Omega^1(M),\mathfrak{spin}(p,q))$.
\end{definition}
If you're not sure why this gives a connection we recommend looking at \cite{MathGauge} and \cite{Pfb}.

We can now notice that the condition of $SO^+(p,q)$-equivariance carries on from one connection to the other:

\begin{proposition}
The spin connection is $\lambda(SO^+(p,q))$ equivariant if and only if $B_{SO}$ is $SO^+(p,q)$ equivariant.
\end{proposition}
\begin{proof}
This comes directly from proposition \ref{morphisms}.
\end{proof}

\begin{example}[The electroweak and weak force]
If we take $(p,q)=(1,3)$ then $Spin^+(p,q)=SL_2(\mathbf{C})$ and represents the electroweak force. This means that all result that we will have for $SO^+(1,3)$ will apply for the electroweak force.

If we take $n=3$ then $Spin(n)=SU(2)$ which represents the weak interaction. We will go further in depth into this in section \ref{weak}.
\end{example}

To evaluate force, and fields from potentials, which are represented by our connections, we need to look at their induced curvature.
\subsection{Curvature}
To introduce curvature, we need to see that since our connection gives rise to parallel transport, it also allows for a derivative which is the one we will use on the connection form itself to get the curvature. We recommend looking at \cite{MathGauge} and \cite{Pfb} for all the details. We will simplify in our context to go faster.
\begin{definition}[Exterior derivative induced by a connection $1$-form]
Our connection $1$-form induces an \textbf{exterior derivative} defined by:
$$
d_B \omega = d \omega + \frac{1}{2}[B\wedge \omega]
$$
Where the $dx^\mu\wedge dx^{\mu_1} \wedge... \wedge dx^{\mu_k}$ component of $[B\wedge\omega]$ is given by the Lie algebra product acting on the components $[B_{\mu},\omega_{\mu_1,...,\mu_k}]$.
\end{definition}
\begin{proposition}[Curvature $2$-form]
The curvature $2$-form through a section bundle covering of $M$ or field strength, $F\in \Omega^2(M,\mathfrak{g})$ is related to its principal connection form by
$$
\begin{aligned}
F_{\alpha,\beta}&=d_B B\\
&=\frac{\partial B_{\beta}}{\partial x^{\alpha}} - \frac{\partial B_{\alpha}}{\partial x^{\beta}} + [B_{\alpha},B_{\beta}]
\end{aligned}
$$
Where $[.,.]$ is the Lie algebra bracket.
\end{proposition}

\newpage

\section{The right left action}
We want to define an appropriate local action from a subgroup of $SO^+(p,q)$ to $M$. First, we need to know where we want to land. This section takes a lot from chapter three of \cite{SRG}. Let us fix $O \in M$ and suppose $M$ to be connected.

\begin{definition}[Group of isometries]
Let us call $(M,g)$'s \textbf{group of isometries} the following:
$$
\begin{aligned}
SO^+(g):=\{&\text{diffeomorphisms } \varphi : M \rightarrow M \\
 & \varphi(O)=O, \, \varphi_* g=g,\, \varphi
\text{ preserves time orientation}\}
\end{aligned}
$$
\end{definition}
Intuitively we want our group action to land here. To do so we turn ourselves toward killing vectors. Let us reintroduce killing vectors in a way that is more adapted to our context.

Since we are looking for a Lie Group, there must exists smooth paths $\varphi_{\epsilon}$ landing in $SO^+(g)$. Furthermore, such smooth paths are charaterized by a fixed point (here $O$) and their derivative:
$$
\frac{d}{d\epsilon} \varphi_{\epsilon} (x)= \xi(\varphi_{\epsilon}(x))
$$
We suppose their derivative to only depend on the value of the function because that is how it works with Lie Groups. We therefore interest ourselves in the flows of vector fields $\xi$ such that they land in $SO^+(g)$.

Let us begin with a few non trivial properties on such flows.

\begin{proposition}[Flows]
Given $\xi$ a vector field of $M$ and $x=(x_1,...,x_n)$ a local atlas we have:
\begin{enumerate}
    \item $Fl_{\alpha\,\epsilon}^{\xi}=Fl_{\epsilon}^{\alpha \, \xi }$
    \item $d_p Fl_{\epsilon}^{\xi}=Fl_{\epsilon}^{d_{Fl_{\epsilon}^{\xi}(p)} \xi}$
    \item In local coordinates: $Jac_x Fl_{\epsilon}^{\xi}=Fl_{\epsilon}^{Jac_{Fl_{\epsilon}^{\xi}(x)} \xi}$
\end{enumerate}
\end{proposition}
\begin{proof}
Firstly, notice that given $p\in M$:
$$
\frac{d}{d\epsilon}Fl_{\alpha\,\epsilon}^{\xi}(p)=\alpha\,\xi(Fl_{\alpha\,\epsilon}^{\xi}(p))
$$

Secondly, thanks to Cauchy, let us compute:
$$
\begin{aligned}
\frac{d}{d\epsilon} d_p Fl_{\epsilon}^{\xi}&=d_p (\frac{d}{d\epsilon} Fl_{\epsilon}^{\xi}) \\
& = d_p (\xi(Fl_{\epsilon}^{\xi})) = (d_{Fl_{\epsilon}^{\xi}(p)} \xi) \circ (d_p Fl_{\epsilon}^{\xi})
\end{aligned}
$$
Which is exactly what we wanted to show.
\end{proof}

With this, we can see how such vectors that conserve the metric look like:

\begin{proposition}[Killing Vectors]
We call a vector field $\xi$ of $M$ a non translating \textbf{killing vector} if its flow lands in $SO^+(g)$. In local coordinates, this is equivalent to $\xi(O)=0$ and
$$
(Jac_x \xi)^T\cdot g_x+g_x\cdot (Jac_x \xi) + d_x g(\xi(x)) = 0
$$
Or as physicist would write it:
$$
\nabla_{\rho}\xi_{\nu} + \nabla_{\nu} \xi_{\rho}=0
$$
We will note $\mathfrak{k}_O$ the set of such vector fields.
\end{proposition}
\begin{proof}
Firstly, since $\xi(O)=0$, by definition, $\frac{d}{d\epsilon} Fl_{\epsilon}^{\xi}(O)=0$ and therefore $ Fl_{\epsilon}^{\xi}(O)=O$.

Furthermore, since $d_p Fl_{0}^{\xi}=Id \in SO^+(g)$ and that the time preservation condition derives from the smoothness in $\epsilon$ of the flow. We just have to check that $(Fl_{\epsilon}^{\xi}) _* g=g$. To do so we only need to check what happens near $\epsilon=0$ at all points because:
$$
(Fl_{\epsilon+\eta}^{\xi})_* g=(Fl_{\epsilon}^{\xi})_* (Fl_{\eta}^{\xi})_* g
$$
Thanks to basic properties of flows. We then see that we need:
$$
\frac{d}{d\epsilon}((Fl_{\epsilon}^{\xi})^* g)=0
$$
Which is equivalent at $\epsilon=0$ in local coordinates to :
$$
(Jac_x \xi)^T\cdot g_x+g_x\cdot (Jac_x \xi) + d_x g(\xi(x)) = 0
$$
\end{proof}

In the following we will need proposition 62 of chapter 3 from \cite{SRG}.

\begin{proposition}
If $M$ is connected and $\phi$ and $\psi$ are in $SO^+(g)$ such that:
$$
d_O \phi= d_O \psi
$$
then $\phi=\psi$.
\end{proposition}

\begin{theorem}[The Special Orthornormal Orthochronous Group of $g$]
If we have fixed tetrads $\sigma$ such that $g=\sigma^T \cdot I_{p,q} \cdot \sigma$. Then there exists a Lie subgroup $K_O$ of $SO^+(p,q)$ which we will call \textbf{special orthornormal orthochronous group of $g$} such that:
$$
\begin{aligned}
\psi_O : \mathfrak{k}_O & \longrightarrow K \\
 \xi & \longmapsto \sigma(O)\cdot Fl_1^{-Jac_O \xi}\cdot \sigma(O)^{-1} = \sigma(O)\cdot exp(-Jac_O \xi) \cdot \sigma(O)^{-1}
\end{aligned}
$$
defines a surjective exponential for $K_O$.
\end{theorem}

\begin{proof}
Let us do this step by step:
\begin{enumerate}
    \item First, let us build an injective Lie algebra morphism:
    $$
    \begin{aligned}
    \varphi_O: \mathfrak{k}_O & \longrightarrow \mathfrak{so}(p,q) \\
    \xi & \longmapsto -\sigma(O)\cdot Jac_O \xi \cdot \sigma(O)^{-1}
    \end{aligned}
    $$
    We can verify it lands in $\mathfrak{so}(p,q)$ by using the killing vector relation in $O$ which gives us:
    $$
    \begin{aligned}
    & (Jac_O \xi)^T\cdot g_O+g_O\cdot (Jac_O \xi) + d_O g(\xi(O)) = 0 \\
    \Leftrightarrow & (Jac_O \xi)^T\cdot g_O+g_O\cdot (Jac_O \xi)=0
    \end{aligned}
    $$
    Since $\xi(O)=0$. We then use $g_O=\sigma(O)^T\cdot I_{p,q} \cdot \sigma(O)$ to get the wanted result.
    
    This is also injective due to  since it fixes the derivatives of $\xi$ at $O$.
    
    We just have to check that it commutes with the Lie algebra product:
    Notice that one can write:
    $$
    [\xi_1,\xi_2]_p=d_p \xi_2 (\xi_1(p)) - d_p \xi_1 (\xi_2(p))
    $$
    and therefore:
    $$
    d_p [\xi_1,\xi_2]=d_p \xi_2 \circ d_p \xi_1 - d_p \xi_1 \circ d_p \xi_2 + d^2_p \xi_2 (\xi_1(p), .) - d^2_p \xi_1 (\xi_2(p),.)
    $$
    Meaning that since $\xi_1(O)=\xi_2(O)=0$:
    $$
    \begin{aligned}
    Jac_O([\xi_1,\xi_2])&=Jac_O\xi_2 \cdot Jac_O \xi_1 - Jac_O\xi_1 \cdot Jac_O \xi_2 \\
    &= -[Jac_O\xi_1, \cdot Jac_O \xi_2]
    \end{aligned}
    $$
    and therefore :
    $$
    \varphi_O([\xi_1,\xi_2])=[\varphi_O(\xi_1),\varphi_O(\xi_2)]
    $$.
    \item Secondly, we can notice that $\psi_O=exp\circ \varphi_O$ and since, as shown in \cite{MatrixGroup}, the exponential is bijective from $\mathfrak{so}(p,q)$ to $SO^+(p,q)$ we have proven the theorem.
\end{enumerate}
\end{proof}

\begin{remark}
The minus sign in the application in necessary for multiplication order to coincide.
\end{remark}

$K$ represents the symmetries that are present in the metric $g$. We can now define the action of $K_O$ on $(M,g)$.

\begin{proposition}
The special orthonormal orthochronous group of $g$ locally acts on the right of $M$ with:
$$
\begin{aligned}
R: M\times K_O & \longrightarrow M \\
(x,\Lambda) & \longmapsto Fl_1^{\psi_O^{-1}(\Lambda)} (x)
\end{aligned}
$$
\end{proposition}

\begin{remark}
If we are placed in Minkowski space-time, this coincides with $L_{\Lambda^{-1}}$.
\end{remark}

\begin{proof}
All we have to check is that $Fl_1^{\psi_O^{-1}(\Lambda_1)}\circ Fl_1^{\psi_O^{-1}(\Lambda_2)}=Fl_1^{\psi_O^{-1}(\Lambda_2\cdot \Lambda_1)}$. To do so we will use proposition 62 of chapter 3 from \cite{SRG} which says that if two local isometries on a connected manifold have the same values in at a point and in their derivatives then they are equal. This is the case and we therefore have a well defined right action.
\end{proof}

\begin{remark}[Another point of view]
Using proposition 62 of \cite{SRG}, we can also see the right action $R_{\Lambda}$ as the only isometry such that $Jac_O R_{\Lambda}=\sigma(O)^{-1}\cdot \Lambda^{-1} \cdot \sigma(O)$. We prefer to pass through killing vectors as it gives us an explicit Lie algebra of our connected Lie group $K_O$.
\end{remark}

With the action in mind we can now consider $G$-equivariant connections where $G$ is a subgroup of $K_O$.

\begin{theorem}[The general semi-Riemannian $G$-equivariant connection]
Given a connected Riemannian manifold $(M,g)$ and $(P,\pi,G)$ a principal fiber bundle where $G$ is a Lie subgroup of $K_O$. With the right action we have defined, a local description $B$ of a principal connection that has a $G$-equivariant section bundle atlas such that $x(O)=0$ satisfies for all $x\in U$ and $\Lambda \in G$:
$$
B(L_{\Lambda}(x))=Ad(\Lambda)(B(x))\circ (T_x L_{\Lambda})^{-1}
$$
And if we define $\Tilde{\Lambda}$ by $T_x L_{\Lambda} = \sigma(L_{\Lambda} (x))^{-1}\cdot \Tilde{\Lambda}(x) \cdot \sigma(x)$ where $\Tilde{\Lambda}(0)=\Lambda$. We then have the following for all $x$:
\begin{enumerate}
    \item $\Tilde{\Lambda}(x)\in SO^+(p,q)$.
    \item $(\Lambda_1 \cdot \Lambda_2)^{\sim}(x)=\Tilde{\Lambda}_1(L_{\Lambda_2}(x)) \cdot \Tilde{\Lambda}_2(x)$
\end{enumerate}
And we can take the more appropriate object:
$$
\omega_{\mu}(x)=\sum_{\nu=1}^n [\sigma^{-1}(x)]_{\nu,\mu}\,B_{\nu}(x)
$$
which satisfies
$$
\omega_{\mu}(L_{\Lambda}(x))=\sum_{\nu=1}^n [\Tilde{\Lambda}(x)^{-1}]_{\nu,\mu} \, Ad(\Lambda)(\omega(x))_{\nu}
$$
\end{theorem}

\begin{remark}
In the special cases where $\Tilde{\Lambda}$ is constant, we come back to an algebraically $K_O$-equivariant $1$-form.
\end{remark}

With that in mind we can apply our symmetries to reduce Yang-Mills equations.
\newpage

\section{Yang-Mills}
Yang-Mills equations are a special case of Euler-Lagrange equations, let us introduce the necessary lagrangian formalism before going any further.
 \subsection{Lagrangian Formalism}
 
 Firstly we need to define what a Lagrangian is in our context. This mixes the calculus of variation and differential geometry.
 \begin{definition}[Lagrangian]
 A \textbf{Lagrangian} on an $n$-Manifold $M$ equipped with a principal bundle $(P, \pi,G)$ is an $n$-form $L$ depending on the connection (or connection $1$-form) and its curvature (or curvature matrix) composed with a fixed set of local sections. ($L=L(B)$)
 \end{definition}
 
 Now we'll associate an action to said Lagrangian
 \begin{definition}[Action]
 Given a Lagrangian $L$, its \textbf{action} is defined as 
 $$
 S=\int_M L 
 $$
 Where we choose the volume form accordingly.
 \end{definition}
 This allows us to define the Lagrangian problem.
 \begin{problem}[Lagrangian optimisation or Principle of Least Action]\label{least action}
 What connection gives us minimal action ?
 \end{problem}
 
 \begin{definition}[Euler Lagrange Equations]
 The differential equations that arise while solving the Lagrangian problem are called the \textbf{Euler-Lagrange equations}.
 \end{definition}

As is often the case with variational calculus, we will need something to act as an integration by part. If we fix a metric on our manifold then we can look at the following operator.

\subsection{The Hodge Operator}
 First we'll introduce a duality between $k$-forms spaces as seen in \cite{Elec},\cite{Pfb} and \cite{MathGauge}.
 \begin{definition} [Hodge Operator]
 If we take $\omega \in \Omega^n (M)$ and $\langle.|.\rangle$ an inner product on $\Omega^k (M)$ then we can define for $\alpha \in \Omega^k (M)$
 $$
     *\alpha \in \Omega^{n-k}(M)
$$
 Such that for all $\beta \in \Omega^k(M)$ we have $\beta\wedge(*\alpha)=\langle \alpha|\beta\rangle \,\omega $
 \end{definition}
 We see that we therefore need an inner product $\langle.|.\rangle$ on $\Omega(M)$.
 \begin{proposition}
 An inner product $\langle.|.\rangle$ on $\Omega^1 (M)$ induces one on $\Omega^k (M)$ for all $k$.
 \end{proposition}
 \begin{proof}
 We can just pick the inner product defined such that 
 
 $$\langle dx_I|dx_J \rangle=\sum_{\sigma \in \Sigma_k} \varepsilon(\sigma)\langle dx_{i_1}|dx_{j_{\sigma(1)}}\rangle...\langle dx_{i_k}|dx_{j_{\sigma(k)}}\rangle$$
 \end{proof}
 
 \begin{proposition}[On Pseudo-Riemanninan Manifolds]
 Given a metric $g=g_{\mu,\nu}$ on $M$, if we write its inverse $g^{\mu,\nu}$ then it induces an inner product on $\Omega^1(M)$ by:
 $$
 \langle dx_i | dx_j \rangle= g^{i,j}
 $$
 \end{proposition}
 
 \begin{proposition}[Explicit formula for the Hodge Operator]
 If we are given a metric $g$ and take $\omega=\sqrt{|g|}\,dx^1 \wedge ... \wedge dx^n $ (the associated volume form) then the Hodge operator given by the inner product induced by $g$ is given by:
 $$
*(dx^{\mu_1}\wedge dx^{\mu_2} ... \wedge dx^{\mu_k})=\frac{\sqrt{|g|}}{(n-k)!}\,\sum_{\nu_1,...,\nu_n=1}^n \epsilon_{\nu_1,...,\nu_n}\,g^{\nu_1,\mu_1}\, ...g^{\nu_k,\mu_k}\,dx^{\nu_{k+1}}\wedge ...\wedge dx^{\nu_n}
 $$
 \end{proposition}
 \begin{proof}
 Simply see that given $dx^{\alpha_1}\wedge ... \wedge dx^{\alpha_k} \in \Omega(M,\mathfrak{g})$ we get:
 $$
 \begin{aligned}
 dx^{\alpha_1}\wedge ... \wedge dx^{\alpha_k}\wedge *(dx^{\mu_1}\wedge dx^{\mu_2} ... \wedge dx^{\mu_k}) &= \frac{\sqrt{|g|}}{(n-k)!}\,\sum_{\nu_{k+1},...,\nu_n=1}^n \epsilon_{\alpha_1,...\alpha_k,\nu_{k+1},\nu_n}\,g^{\alpha_1,\mu_1}\, ...g^{\alpha_k,\mu_k}\\
 &dx^{\alpha_1}\wedge ... \wedge dx^{\alpha_k}\wedge dx^{\nu_{k+1}}\wedge ... \wedge dx^{\nu_n} \\
 &= \sqrt{|g|}\,g^{\alpha_1,\mu_1}\, ...g^{\alpha_k,\mu_k}\, dx^1 \wedge ... \wedge dx^n
 \end{aligned}
 $$
 \end{proof}
 
 We now can see how Lie algebra morphisms commute with the Hodge and differential:
 \begin{proposition}
If $\lambda:\mathfrak{g}\longrightarrow \mathfrak{h}$ is a smooth Lie algebra morphism then, taking $B$ a local description of a connection of a $G$ principle bundle we have the following:
\begin{itemize}
    \item $*\circ\lambda=\lambda \circ *$
    \item $\lambda\circ d_B=d_{\lambda\circ B} \circ \lambda$
\end{itemize}
\end{proposition}
\begin{proof}
The first claim of the proposition is immediate as the Hodge operator does not affect the lie algebra part of $1$-forms and $\lambda$ only affects the lie algebra part.

To prove the second claim, take $\omega \in \Omega(M,\mathfrak{g})$, we have:
$$
\begin{aligned}
\lambda\circ d_B \omega &= \lambda\circ (d(\omega) + \frac{1}{2}\,[B\wedge \omega]) \\
&=d(\lambda\circ \omega) + \frac{1}{2}\,[(\lambda\circ B)\wedge (\lambda \circ \omega) ] \\
&=d_{\lambda\circ B}(\lambda\circ \omega)
\end{aligned}
$$
\end{proof}
 
We are now equipped to look at Yang-Mills equations.
 
 \subsection{Yang Mills Equations}
\begin{definition}[Yang Mills Lagrangian and Action]
Let $(P,\pi,M,G)$ be a principal bundle of rank $k$ over an $n$-dimensional manifold $M$. Let $B$ be a connection $1$-form for $P$ with respect to a section bundle covering. Let $F$ be the corresponding curvature $2$-form. Let $\langle .|.\rangle$ be an $Ad(.)$ invariant scalar product on $\Omega^1(M,\mathfrak{g})$. Then \textbf{the Yang Mills Lagrangian} is
$$
L(B)= \frac{1}{4}\, \langle F|F \rangle
$$
Its associated \textbf{Yang Mills Action} is
$$
S(B)=\int_M L(B)
$$
\end{definition}

A connection that is a critical point of the associated Lagrangian problem is called a \textbf{Yang Mills connection} and the corresponding Euler-Lagrange equations are called the \textbf{Yang Mills equations}.

\begin{example}
In the case of Matrix groups, one can take the $Ad(.)$ invariant product on $\mathfrak{g}$ to be:
$$
\langle A,B \rangle = Trace(A\,B)
$$

This then induces an $Ad(.)$ invariant product on $\Omega^k(M,\mathfrak{g})=\mathfrak{g} \otimes \Omega^k(M)$ given a metric $g$ with $\langle .,. \rangle \otimes g$.
\end{example}

\begin{remark}
If we compare this to electromagnetism, this corresponds to a free system with no current or charge.
Furthermore, we still find that curvature is the force field of the system.
\end{remark}
 \begin{theorem} \label{d_b}
 A connection $1$-form $B$ is solution to problem \ref{least action} if and only if:
 $$*d_B*F=0$$
 \end{theorem}
 \begin{proof}
 See \cite{MathGauge} and \cite{Pfb} for the full details but the basic idea is this: If we derive at $0$ we get $\frac{dF_{B+\epsilon\,H}}{d\epsilon}=d_B H$.
 Then $\int \langle d_A \omega| \eta \rangle=\int * \eta \wedge d_A \omega = \pm \int d_A *\eta \wedge \omega \mp \int d_A (*\eta\wedge \omega) = \pm \int \langle *d_A * \eta , \omega \rangle$ if we suppose $M$ to be compact (we can then generalize). Afterward $\frac{dL(B+\epsilon\,H)}{d\epsilon}=2\,\int \langle F_B,d_B H\rangle = 2\,\int \langle *d_B* F_B,H\rangle$ for all $H \in \Omega^1(M,\mathfrak{g})$. This needs to be $0$ since we are looking for connections minimising the Lagrangian, therefore $*d_B* F_B=0$.
 \end{proof}
 
We will now follow \cite{equivariance} and \cite{MathGauge} to simplify the equation.
\begin{proposition}[First rewriting]\label{1st}
Yang mills in the context of a section bundle covering with a diagonal metric $g$ can be rewritten as, for all $\alpha \in \{1,...,n\}$
$$
\sum_{\beta=1}^n g^{\beta,\beta}\ ([B_{\beta},F_{\alpha,\beta}] + \frac{\partial F_{\alpha,\beta}}{\partial x^{\beta}} + F_{\alpha,\beta}(\frac{\partial}{\partial x^{\beta}} \log(\sqrt{|g|}\,g^{\beta,\beta}\,g^{\alpha,\alpha}))=0
$$
\end{proposition}
\begin{proof}
We can start by using theorem \ref{d_b} to get $*d_B*F=0$. Using $F=\frac{1}{2}\,\sum_{\alpha,\beta=1}^n F_{\alpha,\beta}\,dx^{\alpha}\wedge dx^{\beta}$. We can compute component by component before adding them all up. Without writing the sums on $\nu_3,...,\nu_n$ explicitly we get:
$$
\begin{aligned}
d_B * (F_{\alpha,\beta}\,dx^{\alpha} \wedge dx^{\beta})=&\frac{\sqrt{|g|}}{(n-2)!}\,\epsilon_{\alpha,\beta,\nu_3,...,\nu_n}\,g^{\alpha,\alpha}\,g^{\beta,\beta} \\
&[(\frac{\partial F_{\alpha,\beta}}{\partial x^{\alpha}} + [B_{\alpha},F_{\alpha,\beta}])\,dx^{\alpha}\wedge dx^{\nu_3} ... \wedge dx^{\nu_n} \\
& + (\frac{\partial F_{\alpha,\beta}}{\partial x^{\beta}} + [B_{\beta},F_{\alpha,\beta}])\,dx^{\beta}\wedge dx^{\nu_3} ... \wedge dx^{\nu_n}]\\
&+ \frac{\epsilon_{\alpha,\beta,\nu_3,...,\nu_n}}{(n-2)!}\,F_{\alpha,\beta} \\
&\times [\frac{\partial}{\partial x^{\alpha}} (\sqrt{|g|}\,g^{\alpha,\alpha}\,g^{\beta,\beta})\,,dx^{\alpha}\wedge dx^{\nu_3} ... \wedge dx^{\nu_n} \\
&+\frac{\partial}{\partial x^{\beta}} (\sqrt{|g|}\,g^{\alpha,\alpha}\,g^{\beta,\beta})\,,dx^{\beta}\wedge dx^{\nu_3} ... \wedge dx^{\nu_n}]
\end{aligned}
$$
Then we can compute the Hodge operator of each terms:
$$
\begin{aligned}
&*(\frac{\sqrt{|g|}}{(n-2)!}\,\epsilon_{\alpha,\beta,\nu_3,...,\nu_n}\,g^{\alpha,\alpha}\,g^{\beta,\beta}\,dx^{\alpha}\wedge dx^{\nu_3} ... \wedge dx^{\nu_n})\\
&=(-1)^n\,g^{\alpha,\alpha} \,dx^{\beta}
\end{aligned}
$$
$$
\begin{aligned}
&*(\frac{\epsilon_{\alpha,\beta,\nu_3,...,\nu_n}}{(n-2)!}\,\frac{\partial}{\partial x^{\alpha}} (\sqrt{|g|}\,g^{\alpha,\alpha}\,g^{\beta,\beta})\,dx^{\alpha}\wedge dx^{\nu_3} ... \wedge dx^{\nu_n}) \\
&=(-1)^n\,g^{\alpha,\alpha}\,\frac{\partial}{\partial x^{\alpha}}(\log(\sqrt{|g|}\,g^{\alpha,\alpha}\,g^{\beta,\beta})\,dx^{\beta}
\end{aligned}
$$
Summing it all up we get the wanted formula for all $\alpha$.
\end{proof}

\begin{corollary}
If $g$ is constant and diagonal we get that Yang Mills reduces to
$$
\sum_{\beta=1}^n g^{\beta,\beta}\ ([B_{\beta},F_{\alpha,\beta}] + \frac{\partial F_{\alpha,\beta}}{\partial x^{\beta}})=0
$$
\end{corollary}

We are now ready to see what our group equivariant simplifications bring to the table.

\newpage

\section{Reducing the Minkowski Yang-Mills equations}

\subsection{$SO(n)$ equivariant Identity Yang-Mills}

In our case $g=Id$ is diagonal and constant meaning we get Yang mills as:
$$
\sum_{\beta=1}^n [B_{\beta},F_{\alpha,\beta}] + \frac{\partial F_{\alpha,\beta}}{\partial x^{\beta}} =0
$$

This was the initial motivation behind trying to prove the results in \cite{equivariance}, since this reduces to a one dimensional problem. This is already done in \cite{equivariance} and similarly in \cite{Roland}, \cite{Roland2} so I'll just give the result.
\begin{theorem}[$SO(n)$-equivariant Yang Mills equations]
If we consider, for $n\geq 5$, a connection on a $SO(n)$ principal bundle of $\mathbf{R}^n$ with an equivariant section bundle covering then the Yang Mills equations can be rewritten as
$$
g''+(n+1)\,\frac{g'}{r} + (n-2)\,g^2\,(3-r^2\,g)=0
$$
Where $g$ is given by theorem \ref{SO(n)-reduction}.
\end{theorem}

\subsubsection{$SO(4)$-equivariant Yang-Mills}
In this subsection we suppose that $n=4$, $G=SO(4)$ and that the section bundle covering is $SO(4)$-equivariant. We will also define another notation for $\mathfrak{so}(4)$'s basis.
$$
[\overline{e_{k,l}}]_{i,j}=\epsilon_{k,i,j,l}
$$
and
$$
e_{i,j}=e^A_{i,j}
$$
Then some useful properties on our basis will be their commutation relations.

\begin{lemma}\label{calc}
We have the following formulas:
\begin{itemize}
    \item $[e_{k,\beta},e_{l,\alpha}]=[\overline{e_{k,\beta}},\overline{e_{l,\alpha}}]=\delta_{\alpha}^k\,e_{\beta,l} + \delta_{\alpha}^{\beta}\,e_{l,k}+\delta_{l}^{\beta}\,e_{k,\alpha}+\delta_{l}^k\,e_{\alpha,\beta}$
    \item $[e_{k,\beta},e_{l,\beta}]=[\overline{e_{k,\beta}},\overline{e_{l,\beta}}]=e_{l,k}+\delta_{\beta}^k\,e_{\beta,l} + \delta_{\beta}^l\,e_{k,\beta}$
    \item $[e_{k,\beta},\overline{e_{l,\alpha}}]=\sum_{m=1}^4 \epsilon_{l,\beta,m,\alpha}\,e_{k,m} + \epsilon_{l,m,k,\alpha}\,e_{\beta,m}$
    \item $[e_{k,\beta},\overline{e_{l,\beta}}]=\sum_{m=1}^4  \epsilon_{l,m,k,\beta}\,e_{\beta,m}$
    \item $\sum_{\beta=1}^4 \epsilon_{i,j,k,\beta} \, \epsilon_{l,m,\alpha,\beta}= \delta_i^l\,(\delta_j^m\,\delta_k^{\alpha}-\delta_j^{\alpha}\,\delta_k^m)+\delta_i^m\,(\delta_j^{\alpha}\,\delta_k^l-\delta_j^l\,\delta_k^{\alpha})+\delta_i^{\alpha}\,(\delta_j^l\,\delta_k^m-\delta_j^m\,\delta_k^l)$
\end{itemize}
\end{lemma}
\begin{proof}
Computations to do by hand.
\end{proof}

\begin{proposition}
The curvature form can be rewritten as
$$
\begin{aligned}
F_{\alpha,\beta}(x) = & (f^2 + g^2 + \frac{g'}{r})\,\sum_{k,l,m=1}^4 x_k\,x_l\,(\delta_l^{\alpha}\,\delta_m^{\beta}-\delta_m^{\alpha}\,\delta_l^{\beta})\,e_{k,m} \\
& + (2\,g-r^2\,(g^2+f^2))\,e_{\alpha,\beta} \\
& + \frac{f'}{r}\,\sum_{k,l,m=1}^4 x_k\,x_l\,(\delta_l^{\alpha}\,\delta_m^{\beta}-\delta_m^{\alpha}\,\delta_l^{\beta})\,\overline{e_{k,m}} \\
& + 2\,f\,\overline{e_{\alpha,\beta}} \\
& + 2\,f\,g\,\sum_{k,l,m=1}^4 \epsilon_{k,l,\alpha,\beta} \, x_k \, x_m \, e_{l,m}
\end{aligned}
$$
\end{proposition}
\begin{proof}
First let us use what we have proven in \ref{SO(n)-reduction} to rewrite
$$
B_{\mu}(x)= f(|x|) \, \sum_{k=1}^{4} x_k \, \overline{e_{k,\mu}} + g(|x|)\,\sum_{k=1}^4 x_k \, e_{k,\mu}
$$
Then we can compute each of $F$'s terms.
\begin{enumerate}
    \item $$
    \begin{aligned}
    \frac{\partial B_{\beta}}{\partial x^{\alpha}}=& \frac{x_{\alpha}}{r}\,f'\,\sum_{k=1}^{4} x_k \, \overline{e_{k,\beta}} + f\,\overline{e_{\alpha,\beta}} \\
    &+ \frac{x_{\alpha}}{r} \,g'\, \sum_{k=1}^4 x_k \, e_{k,\beta} + g\,e_{\alpha,\beta}
    \end{aligned}
    $$
    and therefore
    $$
    \begin{aligned}
    \frac{\partial B_{\beta}}{\partial x^{\alpha}}-\frac{\partial B_{\alpha}}{\partial x^{\beta}} = & \frac{f'}{r}\,\sum_{k,l,m=1}^4 x^k\,x^l\,(\delta_l^{\alpha}\,\delta_m^{\beta}-\delta_m^{\alpha}\,\delta_l^{\beta})\,\overline{e_{k,m}} \\
    & + 2\,f\,\overline{e_{\alpha,\beta}} + 2\,g\,e_{\alpha,\beta} \\
    & + \frac{g'}{r}\,\sum_{k,l,m=1}^4 x_k\,x_l\,(\delta_l^{\alpha}\,\delta_m^{\beta}-\delta_m^{\alpha}\,\delta_l^{\beta})\,e_{k,m}
    \end{aligned}
    $$
    \item We now just use lemma \ref{calc} to get:
    $$
    \begin{aligned}
    \relax[B_{\alpha},B_{\beta}]= & 2\,f\,g\,\sum_{k,l,m=1}^4 x_k\,x_m\,\epsilon_{k,l,\alpha,\beta}\,e_{l,m} \\
    &+ (f^2+g^2)\,(\sum_{k,l,m=1}^4 x_k\,x_l\,(\delta_l^{\alpha}\,\delta_m^{\beta}-\delta_m^{\alpha}\,\delta_l^{\beta})\,e_{k,m} -r^2\,e_{\alpha,\beta})
    \end{aligned}
    $$
\end{enumerate}
Putting it together, we get the formula.
\end{proof}
All we have to do now is compute the missing terms to get the Yang-Mills equations.
\begin{theorem}[$SO(4)$-equivariant Yang Mills equations]
In our context, the Yang Mills equations simplify to
$$
\begin{cases}
g'' + \frac{5\,g'}{r} + 4\,g^2 + 2\,(1-r^2\,g)\,(g^2-f^2)=0 \\
f'' + \frac{5\,f'}{r} + 4\,f\,g - 2\,r^2\,f\,(g^2-f^2)=0
\end{cases}
$$
\end{theorem}
\begin{proof}
We just have compute the different terms given by the first rewriting in \ref{1st}.
\begin{enumerate}
    \item $$
    \begin{aligned}
    \sum_{\beta=1}^4 \frac{\partial F_{\alpha,\beta}}{\partial x^{\beta}}=& (g'' + \frac{5\,g'}{r} + 2\,(f^2 + g^2))\,\sum_{k=1}^4 x_k \, e_{\alpha,k} \\
    & + (f'' + \frac{5\,f'}{r} + 4\,f\,g)\,\sum_{k=1}^4 x_k\,\overline{e_{\alpha,k}}
    \end{aligned}
    $$
    \item For the terms of the commutator sum $\sum_{\beta=1}^4[B_{\beta},F_{\alpha,\beta}]$ we can compute each of them, one at a time to get:
    \begin{table}[h!]
    \caption{Commutation table for $SO(4)$}
    \begin{center}
        \begin{tabular}{|c||c|c|}
        \hline
        $\sum_{\beta=1}^4 [c,l]$ & $\sum_{k=1}^4 x_k \, \overline{e_{k,\beta}}$ & $\sum_{k=1}^4 x_k \, e_{k,\beta}$  \\
        \hline
        \hline
         $\sum_{k,l,m=1}^4 x_k\,x_l\,(\delta_l^{\alpha}\,\delta_m^{\beta}-\delta_m^{\alpha}\,\delta_l^{\beta})\,e_{k,m}$ & 0 & 0 \\
         \hline
         $e_{\alpha,\beta}$ & $-2\,\sum_{k=1}^4 x_k\,\overline{e_{\alpha,k}}$ & $2\,\sum_{k=1}^4 x_k\,e_{\alpha,k}$ \\
         \hline
         $\sum_{k,l,m=1}^4 x_k\,x_l\,(\delta_l^{\alpha}\,\delta_m^{\beta}-\delta_m^{\alpha}\,\delta_l^{\beta})\, \overline{e_{k,m}}$ & 0 & 0 \\
         \hline
         $\overline{e_{\alpha,\beta}}$ & $2\,\sum_{k=1}^4 x_k\,e_{\alpha,k}$ & $2\,\sum_{k=1}^4 x_k\,\overline{e_{\alpha,k}}$ \\
         \hline
         $\sum_{k,l,m=1}^4 \epsilon_{k,l,\alpha,\beta} \, x_k \, x_m \, e_{l,m}$ & $2\,r^2\,\sum_{k=1}^4 x_k\,e_{\alpha,k}$ & $-2\,r^2\,\sum_{k=1}^4 x_k\,\overline{e_{\alpha,k}}$ \\
         \hline
        \end{tabular}
    \end{center}
    \end{table}
    This gives us
    $$
   \begin{aligned}
    \sum_{\beta=1}^4 [B_{\beta},F_{\alpha,\beta}] =& [2\,g^2-f^2-r^2\,g^3+r^2\,g\,f^2](2\,\sum_{k=1}^4x_k\,e_{\alpha,k}) \\
    & + [-r^2\,f\,g^2-r^2\,f^3](2\,\sum_{k=1}^4 x_k\,\overline{e_{\alpha,k}})
    \end{aligned}
    $$
\end{enumerate}
We just add the two and project onto the almost everywhere linearly free vectors $\sum_{k=1}^4 x_k\,e_{\alpha,k}$ and $\sum_{k=1}^4 x_k\,\overline{e_{\alpha,k}}$ to get the wanted equations.
\end{proof}
\begin{remark}
If we take $f=0$ we fall back to the $n\geq 5$ equation shown in \cite{equivariance}.
\end{remark}

\subsection{$SO^+(p,q)$-equivariant Yang-Mills}
In our case $g=I_{p,q}$ is constant and diagonal so Yang-Mills reduces to:
$$
\sum_{\beta=1}^n \epsilon(\beta)\,([B_{\beta},F_{\alpha,\beta}] + \frac{\partial F_{\alpha,\beta}}{\partial x^{\beta}})=0
$$

\subsubsection{$SO^+(p,q)$-equivariant Yang-Mills for $p+q \geq 5$}

First we will choose a basis for $\mathfrak{so}(p,q)$. We will take $(f_{i,j})_{i<j}=(e^A_{i,j}\cdot I_{p,q})$. We can then rewrite our equivariant local connection $1$-form as :
$$
B_{\mu}(x)=g(r)\,\sum_{k=1}^n x_k\,f_{k,\mu}=g(r)\,X_{\mu}
$$
We can then do the following computations to help us further along.
\begin{lemma}[Computational help for $SO^+(p,q)$]\label{Comso}
We have the following formulas:
\begin{itemize}
    \item $f_{k,\alpha}=-f_{\alpha,k}$
    \item $[f_{k,\alpha},f_{l,\beta}]=\epsilon(\alpha)\,(\delta_l^{\alpha}\,f_{k,\beta} + \delta_{\beta}^{\alpha} \, f_{l,k}) + \epsilon(k)\,(\delta_k^l\,f_{\beta,\alpha}+\delta_{\beta}^k\, f_{\alpha,l})$
    \item
    $\frac{\partial r}{\partial x^{\alpha}}=\epsilon(\alpha)\,\frac{x_{\alpha}}{r}$
\item
    $\frac{\partial X_{\mu}}{\partial x^{\nu}}=f_{\nu,\mu}$
\item
    $\sum_{\beta=1}^n x_{\beta}\,X_{\beta}=0$
\item
    $[X_{\alpha},X_{\beta}]=\epsilon(\alpha)\,x_{\alpha}\,X_{\beta} - r^2\,f_{\alpha,\beta} - \epsilon(\beta)\,x_{\beta}\,X_{\alpha}$
\item
    $[X_{\beta},f_{\alpha,\beta}]=-\epsilon(\beta)\,((1-\delta_{\alpha}^{\beta})\,X_{\alpha}+x_{\beta}\,f_{\alpha,\beta})$
\end{itemize}
\end{lemma}

With that we can compute the curvature:
\begin{proposition}[Curvature of an $SO^+(p,q)$-equivariant connection]\label{p,q,curv}
The curvature associated to an $SO^+(p,q)$-equivariant connection is of the form:
$$
F_{\alpha,\beta}=(\frac{g'(r)}{r}+g(r)^2)\,(\epsilon(\alpha)\,x_{\alpha}\,X_{\beta}-\epsilon(\beta)\,x_{\beta}\,X_{\alpha}) + g(r)\,(2-r^2\,g(r))\,f_{\alpha,\beta}
$$
\end{proposition}
\begin{proof}
This just follows from the computation lemma.
\end{proof}

\begin{theorem}[$SO^+(p,q)$-equivariant Yang Mills equations]
If we consider, for $p+q\geq 5$, a connection on a $SO^+(p,q)$ principal bundle of $\mathbf{R}^n$ with an equivariant section bundle covering then the Yang Mills equations can be rewritten as
$$
g''+(n+1)\,\frac{g'}{r} + (n-2)\,g^2\,(3-r^2\,g)=0
$$
Where $g$ is given by theorem \ref{SO(p,q)-reduction}.
\end{theorem}
\begin{remark}
We get the same reduction as in the $SO(n)$-case and therefore generalise what was done in \cite{equivariance}.
\end{remark}
\begin{proof}
We compute the terms step by step:
$$
\begin{aligned}
\sum_{\beta=1}^n \epsilon(\beta)\,\frac{\partial F_{\alpha,\beta}}{\partial x^{\beta}}&= \sum_{\beta=1}^n \frac{x^{\beta}}{r}\,(\frac{g''}{r} - \frac{g'}{r^2} + 2\,g\,g')\,(\epsilon(\alpha)\,x_{\alpha}\,X_{\beta}-\epsilon(\beta)\,x_{\beta}\,X_{\alpha}) \\
& +(\frac{g'}{r} +g^2)\,(\epsilon(\alpha)\,\epsilon(\beta)\,\delta_{\alpha}^{\beta}\,X_{\beta} - X_{\alpha}-x_{\beta}\,f_{\beta,\alpha}) \\
& + \frac{x_{\beta}}{r}\,(2\,g'-2\,r\,g\,(g+r\,g'))\,f_{\alpha,\beta} \\
&= -(g''-\frac{g'}{r} + 2\,r\,g\,g'+n\,(\frac{g'}{r}+g^2)+\frac{2\,g'}{r}-2\,g\,(g+r\,g'))\,X_{\alpha}\\
&= -(g''+(n+1)\frac{g'}{r} + (n-2)\,g^2)\,X_{\alpha}
\end{aligned}
$$
$$
\begin{aligned}
\sum_{\beta=1}^n \epsilon(\beta)\,[B_{\beta},F_{\alpha,\beta}]&=g\,\sum_{\beta=1}^n (\frac{g'}{r} +g^2)\,(\epsilon(\alpha)\,\epsilon(\beta)\,x_{\alpha}\,[X_{\beta},X_{\beta}] - x_{\beta}\,[X_{\beta},X_{\alpha}]) \\
&+ g\,(2-r^2\,g)\,\epsilon(\beta)\,[X_{\beta},f_{\alpha,\beta}] \\
&= -g^2\,(2-r^2\,g)\,(n-2)\,X_{\alpha}
\end{aligned}
$$
And so adding them gives us the redundant equation we wanted.
\end{proof}

\subsubsection{$SO^+(p,q)$-equivariant Yang-Mills for $p+q=4$}

Let us start by encoding the given simplification:
$$
B=g\,X+f\,Y
$$
Where,
$$
\begin{aligned}
X_{\mu}=\sum_{k=1}^4 x_k\,f_{k,\mu} \\
Y_{\mu}=\sum_{k=1}^4 x_k\,\overline{f}_{k,\mu}
\end{aligned}
$$
Where we introduce:
$$
\overline{f}_{k,\mu}=I_{p,q}\cdot \overline{e}_{k,\mu}
$$
We then have the following useful formulas that add up with the previous ones we found on $X_{\mu}$ in \ref{Comso}.
\begin{lemma}[Computational help for $p+q=4$]
We have the following formulas:

\begin{itemize}
    \item $\epsilon(a)\,\epsilon_{a,b,c,d}=(-1)^p\,\epsilon(b)\,\epsilon(c)\,\epsilon(d)\,\epsilon_{a,b,c,d}$
    \item $[f_{k,\alpha},\overline{f}_{l,\beta}]=\sum_{n=1}^4\epsilon(n)\,(\epsilon_{l,n,\alpha,\beta}\,f_{n,k} + \epsilon_{l,n,k,\beta}\,f_{n,\alpha})$
    \item $[\overline{f}_{k,\alpha},\overline{f}_{l,\beta}]=(-1)^p\,[\epsilon(k)\,\delta_k^l\,f_{\beta,\alpha}+\epsilon(\alpha)\,\delta_{\alpha}^{\beta}\,f_{l,k}+\epsilon(\beta)\,\delta_k^{\beta}\,f_{\alpha,l}+\epsilon(l)\,\delta_l^{\alpha}\,f_{k,\beta}]$
    \item $[Y_{\alpha},Y_{\beta}]=(-1)^p\,[X_{\alpha},X_{\beta}]$
    \item $[X_{\alpha},Y_{\beta}]=\sum_{k,l,n=1}^4 x_k \, x_l \, \epsilon(n)\, \epsilon_{l,n,\alpha,\beta} \, f_{n,k}$
    \item $\frac{\partial Y_{\alpha}}{\partial x^{\beta}}=\overline{f}_{\alpha,\beta}$
\end{itemize}
\end{lemma}
We then compute the curvature:
\begin{proposition}
In our context the curvature is:
$$
\begin{aligned}
F_{\alpha,\beta}=&(\frac{g'}{r} + g^2 + (-1)^p\,f^2)\,(\epsilon(\alpha)\,x_{\alpha}\,X_{\beta}-\epsilon(\beta)\,x_{\beta}\,X_{\alpha}) \\
& + (2\,g-r^2\,(g^2+(-1)^p\,f^2))\,f_{\alpha,\beta} \\
& + \frac{f'}{r}(\epsilon(\alpha)\,x_{\alpha}\,Y_{\beta}-\epsilon(\beta)\,x_{\beta}\,Y_{\alpha})\\
&+ 2\,f\,\overline{f}_{\alpha,\beta} \\
&+ 2\,f\,g\,[X_{\alpha},Y_{\beta}]
\end{aligned}
$$
\end{proposition}
\begin{proof}
The proof is exactly the same as for $SO^+(4)$
\end{proof}

\begin{theorem}[$SO^+(p,q)$-equivariant Yang Mills equations where $p+q=4$]
In our context, Yang Mills equations simplify to
$$
\begin{cases}
g'' + \frac{5\,g'}{r} + 4\,g^2 + 2\,(1-r^2\,g)\,(g^2+(-1)^{p+1}\,f^2)=0 \\
f'' + \frac{5\,f'}{r} + 4\,f\,g - 2\,r^2\,f\,(g^2+(-1)^{p+1}\,f^2)=0
\end{cases}
$$
\end{theorem}
\begin{proof}
Let us compute each component:
\begin{itemize}
    \item $$
    \begin{aligned}
    \sum_{\beta=1}^4 \epsilon(\beta) \, \frac{\partial F_{\alpha,\beta}}{\partial x^{\beta}}=&(g'' + 5\,\frac{g'}{r}+2\,(g^2+(-1)^p\,f^2))(-X_{\alpha}) \\
    &+ (f''+5\,\frac{f'}{r}+4\,f\,g)\,(-Y_{\alpha})
    \end{aligned}
    $$
    \item We have the following commutation table:
    \begin{table}[h!]
    \begin{center}
        \caption{Commutation table for $SO^+(p,q)$ with $p+q=4$}
        \begin{tabular}{|c||c|c|}
        \hline
        $\sum_{\beta=1}^4\epsilon(\beta) [c,l]$ & $X_{\beta}$ & $Y_{\beta}$  \\
        \hline
        \hline
         $X_{\beta}$ & 0 & 0 \\
         \hline
         $\epsilon(\beta)\,x_{\beta}\,X_{\alpha}$& 0 & 0 \\
         \hline
         $f_{\alpha,\beta}$ & $-2\,X_{\alpha}$ & $2\,Y_{\alpha}$ \\
         \hline
         $Y_{\beta}$ & 0 & 0 \\
         \hline
         $\epsilon(\beta)\,x_{\beta}\,Y_{\alpha}$ & 0 & 0 \\
         \hline
         $\overline{f_{\alpha,\beta}}$ & $-2\,Y_{\alpha}$ & $(-1)^p\,2\,X_{\alpha}$ \\
         \hline
         $\sum_{k,l,m=1}^4 \epsilon(l)\,\epsilon_{k,l,\alpha,\beta} \, x^k \, x^m \, f_{l,m}$ & $2\,r^2\,Y_{\alpha}$ & $(-1)^{p+1}\,2\,r^2\,X_{\alpha}$ \\
         \hline
        \end{tabular}
    \end{center}
    \end{table}
    Which gives us:
    $$
    \begin{aligned}
    \sum_{\beta=1}^4 \epsilon(\beta)\,[B_{\beta},F_{\alpha,\beta}] =& [2\,(g^2+(-1)^{p+1}\,f^2-r^2\,g^3+(-1)^p\,r^2\,g\,f^2](-2\,X_{\alpha}) \\
    & + [-r^2\,f\,g^2+(-1)^{p+1}\,r^2\,f^3](-2\,Y_{\alpha})
    \end{aligned}
    $$
\end{itemize}
We then add the two to get the wanted equations as $X_{\alpha}$ and $Y_{\alpha}$ are free almost everywhere.
\end{proof}

\begin{remark}
Notice that once again we can take $f=0$ and find the coherent equation we had in the general case. However, if we take $g=0$ then necessarily $f=0$. To find new self-coherent solutions let us take a look at $p=1$. The equations become:
$$
\begin{cases}
g'' + \frac{5\,g'}{r} + 6\,g^2 + 2\,f^2-2\,r^2\,g\,(g^2+f^2)=0 \\
f'' + \frac{5\,f'}{r} + 4\,f\,g - 2\,r^2\,f\,(g^2+f^2)=0
\end{cases}
$$
We then look for the right linear transformation such that both non derivative terms are multiples of their relative functions. Take for example: $h_1=f+\sqrt{3}\,g$ and $h_2=f-\sqrt{3}\,g$, the equations then become:
$$
h_i'' + \frac{5\,h_i'}{r} + 2\,h_i^2-\frac{2}{3}\,r^2\,h_i\,(h_1^2+h_2^2+h_1\,h_2)=0
$$
Which are self-coherent as you can have non-zero solutions for $h_1=0$ or $h_2=0$.

When $p=0$ or $2$ the answer is a bit more complicated and we get for $h_{\pm}=f\pm g$ :
$$
h_{\pm}'' + \frac{5\,h_{\pm}'}{r}+2\,h_{\pm}\,(h_{\pm} -2\, h_{\mp}) +2\,r^2\,h_{\pm}^2h_{\mp}=0
$$
\end{remark}

\subsection{A conjecture}
Ensuing the results we have proven, we can state the following conjecture:
\begin{theorem}[Conjecture]\label{conj}
Given an $n$ semi-Riemannian manifold $(M,g)$ of signature $(p,q)$ and principal bundle $(P,\pi,M,G)$ such that the Lie group $G$ is injectively represented by $\rho$ in $SO(p,q)$ in a way that it induces a local action on $M$. Then there exists a 'nice' simplification for local $\rho(G)$-equivariant connections that simplify the Yang-Mills Equations.
\end{theorem}

\newpage 

\section{Reducing the Isotropic Yang-Mills Equations}\label{Iso}
In this subsection we want to consider a non-constant $(1,n)$ metric that has spherical spacial symmetry. We will also impose $\partial_t$ to be a killing vector for the metric to be isotropic as in \cite{Isotropic}:
$$
g=\begin{pmatrix}
 e^{f_t(r)} & \\
 & -e^{f_r(r)}\,I_n
\end{pmatrix}
$$
where $r=\sqrt{x^T\cdot x }$.
\begin{remark}
This is the case for the Schwarzschild metric and the Reisner-Nordström metric in isotropic coordinates.
\end{remark}

Here, Yang-Mills is the result of varying our connection while leaving our metric fixed in a more complicated Lagrangian system taking into account General Relativity. The equations determining $f_t$ and $f_r$ in terms of our connection are given by Einstein's equations that are coupled to our system. We will see how far we can go without making any more assumptions on our metric.

\begin{remark}
Notice that contrary to \cite{Gastel2011YangMillsCO} our base manifold is not an $n$-sphere but rather an $n$-ball in space or the entire space itself.
\end{remark}

One can view this as a generalization of what was done in \cite{gu1981spherically}.

\subsection{$SO(n)$-equivariant connections}
Given $B_{SO(n)}(t,x)$ the local description of an $SO(n)=K_{O}$-equivariant connection and the associated $\omega_{\mu}(x)=\sum_{\nu=1}^n [\sigma^{-1}(x)]_{\nu,\mu}\,(B_{SO(n)})_{\nu}(x)$. Then, for $R \in SO(n)$ we get that:
$$
\begin{cases}
\omega_{0}(t,R\cdot x)= Ad(R)(\omega_0(t,x)) \\
\omega_i(t,R\cdot x)=\sum_{j=1}^n [R^{-1}]_{j,i}\,Ad(R)(\omega_i (t,x)) & i\in \{1,...,n\}
\end{cases}
$$

Therefore, $(\omega_i(t,.))_{i\geq 1}$ is an algebraically $SO(n)$-equivariant $1$-form.

\begin{lemma}
Any $X:\mathbf{R}^n\mapsto \mathfrak{so}(n)$ such that $X(R\cdot x)=Ad(R)(X(x))$ is zero.
\end{lemma}
\begin{proof}
As we have done before, we only need to know what happens along $X(r\,e_1)$. Furthermore $X(r\,e_1)$ is invariant by all $Ad(g_{i,j}^{\theta})$ for $i,j \neq 1$.
Since $Ker(Ad(g_{i,j}^{\theta})-Id)=span\{e^A_{i,j},e^A_{k,l}|k,l\notin\{i,j\}\}$, $X(r\,e_1)$ is necessarily $0$.
\end{proof}

And therefore $\omega_0=0$.

\begin{remark}
If we compare this to electromagnetism, it is as if the electric potential was null.
\end{remark}

Using what we know on $SO(n)$-equivariant $1$-forms and putting the coefficients of $\sigma$ inside the functions we get:

\begin{theorem}[$SO(n)$-equivariant $SO(n)$-connection form on a Lorentzian Manifold]
For all $n\geq 3$ the local connection forms we described are of the form:
$$
(B_{SO(n)})_0=0
$$
\begin{enumerate}
    \item For $n=3$
    $$[(B_{SO(3)})_{k}(t,x)]_{i,j}=f(t,|x|)\,\epsilon_{i,j,k} + g(t,|x|)\,( \delta_j^{k}x^i-\delta_i^{k}\,x^j ) + h(t,|x|)\,x^{k}\,\sum_{l=1}^3\epsilon_{i,j,l}\,x^l 
    $$
    \item For $n=4$
    $$
    [(B_{SO(4)})_{k}(t,x)]_{i,j}= f(t,|x|) \, \sum_{l=1}^{4} \epsilon_{l,i,j,\mu} \, x^l + g(t,|x|)\,( \delta_j^{k}x^i-\delta_i^{k}\,x^j)
    $$
    \item For $n\geq 5$
    $$[(B_{SO(n)})_{k}(x)]_{i,j}=g(t,|x|)\,( \delta_j^{k}x^i-\delta_i^{k}\,x^j) 
    $$
\end{enumerate}
\end{theorem}

For simplicity, we will suppose that either $n\geq 5$ or $f,h=0$ and take:
$$
[(B_{SO(n)})_{k}(x)]_{i,j}=g(t,|x|)\,( \delta_j^{k}x^i-\delta_i^{k}\,x^j)=g(t,|x|)\,[X_k]_{i,j}
$$

Using the computations we already made for the Minkowski case we can compute the curvature:
$$
\begin{cases}
(F_{SO(n)})_{i,j}=(\frac{1}{r}\,\frac{\partial g}{\partial r} +g^2)\,(x_i \, X_j - x_j\,X_i) + g\,(2-r^2\,g)\,e^A_{i,j} \\
(F_{SO(n)})_{0,j}=\frac{\partial g}{\partial t}\,X_j
\end{cases}
$$

Then using proposition \ref{1st} Yang-Mills can be written:
$$
\frac{\partial^2 g}{\partial r^2}-e^{-f_t+f_r}\,\frac{\partial^2 g}{\partial t ^2} + (n+1)\,\frac{1}{r}\,\frac{\partial g}{\partial r} + (n-2)\,g^2\,(3-r^2\,g) + \frac{f_t'+(n-4)\,f_r'}{2}\,\left (\frac{\partial g}{\partial r} + 2\,\frac{g}{r} \right )=0
$$
When $\alpha$ is in the spatial coordinates and $0=0$ when it is the time component.

\begin{remark}[Interpretation]
One can interpret this as coupling the source of our gauge field and gravitational field. However, since an $SO(3)$ gauge field does not represent anything in physics we must find a way to lift the results we have found. However, if we consider a spinor (a section of the associated vector bundle into $\mathbf{R}^n$) $|\psi \rangle$ such that there exists $\Lambda \in SO(n)$ and $x\in M$ such that $|\psi \rangle(\Lambda\cdot x)=\Lambda\cdot|\psi \rangle(x) $ we then have the following relation:
$$
\sum_{\nu=1}^n [\Lambda]_{\nu,\mu}\,\frac{\partial |\psi \rangle}{\partial x^{\nu}}(\Lambda\cdot x) = \Lambda\cdot \frac{\partial |\psi \rangle}{\partial x^{\mu}}
$$
If we want to keep this relation with our connection derivative we need:
$$
\sum_{\nu=1}^n [\Lambda]_{\nu,\mu}\,d_B( |\psi \rangle)_{\nu}(\Lambda\cdot x) = \Lambda\cdot d_B(\partial |\psi \rangle)_{\mu}(x)
$$
This is equivalent to:
$$
\begin{aligned}
&\sum_{\nu=1}^n [\Lambda]_{\nu,\mu}\,B_{\nu}(\Lambda\cdot x)\cdot |\psi \rangle(\Lambda\cdot x) = \Lambda\cdot B_{\mu}(x)\cdot |\psi \rangle(x)\\
\Leftrightarrow & \left (\sum_{\nu=1}^n [\Lambda]_{\nu,\mu}\,B_{\nu}(\Lambda\cdot x)\cdot \Lambda- \Lambda\cdot B_{\mu}(x)\right ) \cdot |\psi \rangle( x)=0
\end{aligned}
$$
Having this for all such tuple $(\Lambda, |\psi \rangle ,x)$ is equivalent to $B$ being $SO(n)$-equivariant since for all $y\in \mathbf{R}^n$ we can choose a smooth $|y_{\Lambda,x} \rangle$ such that $|y_{\Lambda,x}\rangle (x)=y$ and $|y_{\Lambda,x} \rangle (\Lambda\cdot x)= \Lambda \cdot y$ whenever $x \notin ker(\Lambda -Id)$, which is almost everywhere.
\end{remark}

\subsection{Lifting to the weak interaction or $Spin(n)$} \label{weak}
In the standard model, the weak interaction is modeled by an $SU(2)$ principal connection (see \cite{MathGauge,Pfb}) that satisfies Yang-Mills. A way to see $SU(2)$ is as the double cover of $SO(3)$ or the group commonly called $Spin(n)$ in the general case. We can take inspiration from spin-structures as seen in \cite{Spinpq} to develop a relationship between a $Spin(n)$ principal fiber bundle and a $SO(n)$ principal fiber bundle to be able to lift up connections.

Firstly, we have a smooth double covering

We will write $\lambda_d$ its differential in $e$ which is a Lie algebra isomorphism by definition.

Let us then suppose that we have $(P_{spin},\pi_{spin},Spin(n))$ the principal fiber bundle that encodes the weak interaction.
We can then introduce the following left action of $Spin(n)$ on $M$
$$
\begin{aligned}
L:Spin(n)\times M & \longrightarrow M \\
(S,(t,x)) & \longmapsto (t,\lambda(S) \cdot x)
\end{aligned}
$$
Which is associated to the left action of $SO(n)$. 

Now, if we take the local description of a principal connection $B_{spin}$ that $Spin(n)$-equivariant with that action, then $B_{so}=\lambda_d\circ B_{spin}\in \Omega^1(M,\mathfrak{so}(n))$ can be seen as the local description of a connection of a non-defined (maybe non existing) $SO(n)$-principal bundle.

\begin{proposition}
If $B_{spin}$ is $Spin(n)$ equivariant for the defined action if and only if $B_{so}$ is $SO(n)$-equivariant.
\end{proposition}
\begin{proof}
Take $R\in SO(n)$ and then take $S\in Spin(n)$ such that $R=\lambda(S)$ then:
$$
\begin{aligned}
(B_{so})_{i}(t,R\cdot x)&=\lambda_d\circ (B_{spin})_i (t,L_S (x)) \\
&=\lambda_d\circ ( \sum_{j=1}^n [R^{-1}]_{j,i}\,Ad(S)((B_{spin})_j(t,x))) \\
&=\sum_{j=1}^n [R^{-1}]_{j,i} Ad(\lambda(S))(\lambda_d\circ (B_{spin})_j(t,x))) \\
&=\sum_{j=1}^n [R^{-1}]_{j,i} Ad(R)((B_{so})_j(t,x))
\end{aligned}
$$
And we have the same for $\mu=0$.
\end{proof}

Furthermore, $B_{so}$ completely describes $B_{spin}$ and thanks to $\lambda_d$ being an isomorphism of Lie algebras the Yang-Mills equations are equivalent.

We can then deduce the following
\begin{corollary}
If we apply the proposition to $\lambda_d$ we get that:
$$
F_{so}=\lambda_d \circ F_{spin}
$$
And $B_{spin}$ satisfies Yang-Mills, if and only if $B_{so}$ satisfies Yang-Mills. 
\end{corollary}

With that in mind, $B_{spin}$ is completely encoded by $B_{so}$ and all the results demonstrated for $SO(n)$ hold for $Spin(n)$.

\begin{remark}[Physical interpretation]
Such a connection can only logically act on a global spinor such that $|\psi(\Lambda\cdot x)\rangle=\Lambda \Tilde{\cdot} |\psi(x)\rangle$ where we need to define the action of $SO(3)$ onto $2$-dimensional complex spinors such that it cancels with the $\Lambda^{-1}$ present in the $Ad(\Lambda)$ on the connection. This type of symmetry does not directly assume \textbf{P} or \textbf{CP} and could therefore be tested. We can then interpret the system we study as the weak interaction field of a particle. The wave equation that arises is therefore the wave characterising the information carried by the particle in the center.
\end{remark}

\subsection{Lifting to the strong interaction or $SU(n)$}
In the case $n=3$, the strong interaction is modeled by $SU(3)$. We can see $SO(n)$ as a Lie subgroup of $SU(n)$ in several ways but the simplest one is the direct injection.

Let us consider $(P_{su},\pi_{su},SU(n))$ the principal fiber bundle that models the strong interactions. Take $B_{su}$ the local description of a connection. We suppose $B_{su}$ to be $SO(n)$-equivariant.

\begin{remark}[Physical interpretation]
This time, it is more straightforward than with the weak interaction, the field would be directly induced by a spinor such that:
$$
\begin{pmatrix}
 \psi_1(\Lambda\cdot x) \\
 \psi_2(\Lambda\cdot x) \\
 \psi_3(\Lambda\cdot x) \\
\end{pmatrix} = \Lambda \cdot \begin{pmatrix}
 \psi_1( x) \\
 \psi_2( x) \\
 \psi_3( x) \\
\end{pmatrix}
$$
And is therefore of the form:
$$
|\psi(x) \rangle = \phi(r)\,\frac{x}{r}
$$
Its global colour charge being neutral, we could interpret it as the spinor of a colour neutral system like a nucleon or the kernel of an atom, where each colour is confined to a direction as below.
\begin{figure}[h!]
    \centering
    \caption{$SO(3)$-equivariant color distribution (in color)}
    \includegraphics[width=0.75\textwidth]{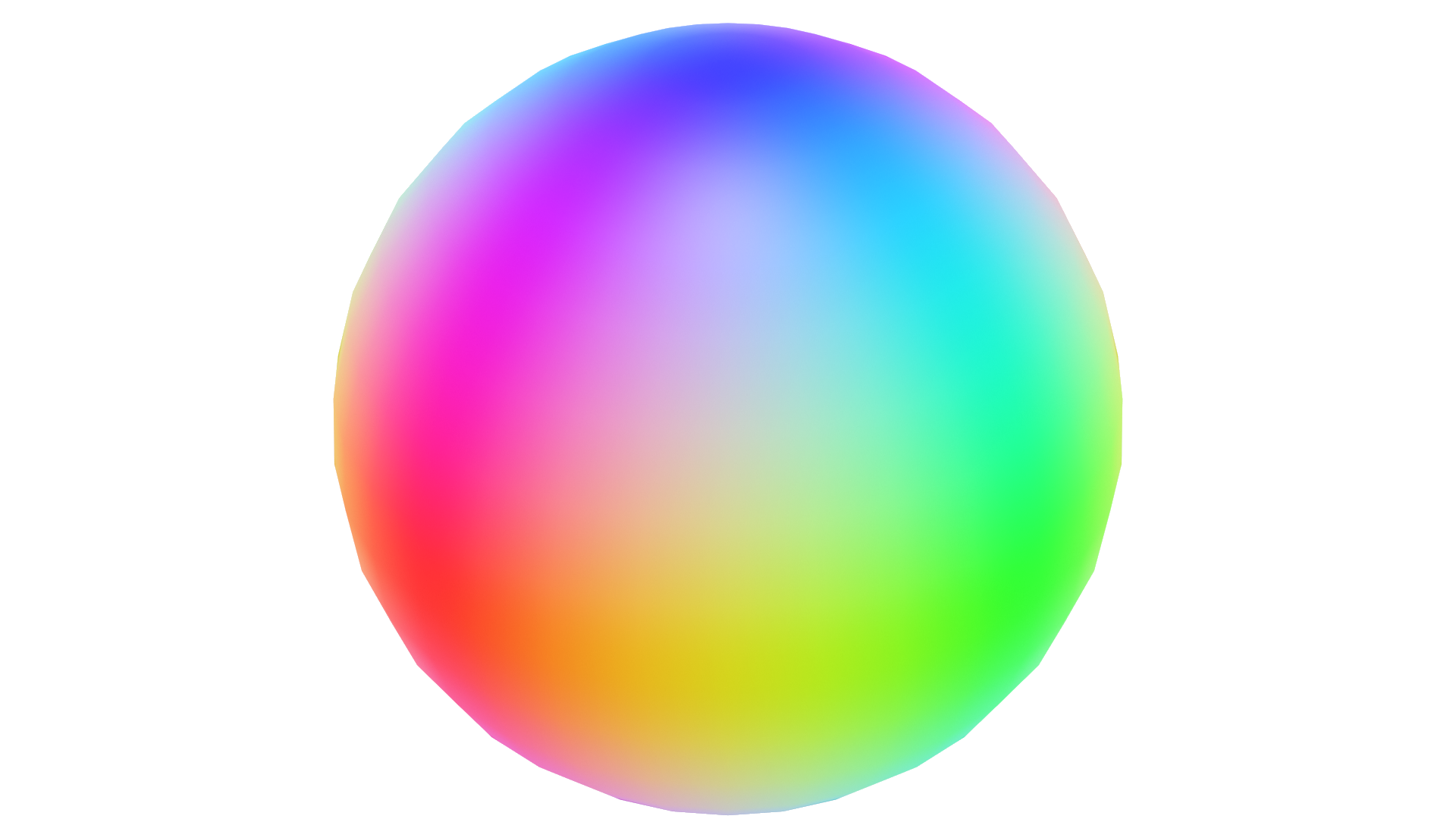}
\end{figure}
\end{remark}

We now need to deduce the general form for such $1$-forms. Let us use what we have done previously:

\begin{proposition}[General form for an $SO(n)$ equivariant $SU(n)$ connection form]
If we impose $n\geq 5$ or certain functions to be null then the general form for an $SO(n)$ equivariant $SU(n)$ connection form is:
$$
\begin{aligned}
(B_{su})_i(t,x)=&f(t,r)\,\left ( x\cdot e_i^T-e_i\cdot x^T \right )  \\
+&i\, \left (g_1(t,r)\,x_i\,\left ( x\cdot x^T-\frac{r^2}{n}\,Id \right ) +g_2(t,r)\,\left ( e_i\cdot x^T + x\cdot e_i^T - \frac{2\,x_i}{n}\,Id \right ) \right ) \\
\end{aligned}
$$
and
$$
(B_{su})_0=0
$$
\end{proposition}
\begin{proof}
As before, we can decompose $B_{su}=B_{so}+i\,S$ where both $B_{so}$ and $S$ are $SO(n)$-equivariant:
\begin{itemize}
    \item We already know how to characterize $B_{so}$ from the previous section.
    \item Through similar reasoning we get that $S_0(t,r\,e_1)$ is $Ad(g_{i,j}^{\theta})$-invariant for all $1<i<j$, meaning that it is proportional to $e_{1,1}$, the null trace condition then forces it to be zero.
    
    We then see that $S_i(t,.)$ is $SO(n)$-equivariant and we use what we have shown in section \ref{SU(n)}.
\end{itemize}
\end{proof}

We now have the general form we can now compute its curvature and Yang-Mills equation.

First, let us put all functions in a dimension-less form. this means we take:

$$
\begin{aligned}
B_i(t,x)=&(h_1(t,r)+1)\,\frac{1}{r^2}\,\left ( x\cdot e_i^T-e_i\cdot x^T \right )  \\
+&i\, \left ((h_3(t,r)-2\,h_2(t,r))\,\frac{x_i}{r^2}\,\left (\frac{1}{r^2} \, x\cdot x^T-\frac{1}{n}\,Id \right ) +h_2(t,r)\, \frac{1}{r^2} \, \left ( e_i\cdot x^T + x\cdot e_i^T - \frac{2\,x_i}{n}\,Id \right ) \right ) \\
=& (h_1+1)\,F_i + h_3\,x_i\,G^1 + h_2\,(G^2_i-2\,x_i\,G^1) 
\end{aligned}
$$
and
$B_0=0$

Where:
$$
F_i=\frac{1}{r^2}\,\left ( x\cdot e_i^T-e_i\cdot x^T \right )
$$
$$
G^1=\frac{i}{r^2}\,\left (\frac{1}{r^2} \, x\cdot x^T-\frac{1}{n}\,Id \right )
$$
$$
G^2_i=\frac{i}{r^2} \, \left ( e_i\cdot x^T + x\cdot e_i^T - \frac{2\,x_i}{n}\,Id \right )
$$

Let us start with a few formulas that we will use throughout the computations:
\begin{lemma}[General computational help for $SU(n)$]
We have the following on our basis:
\begin{itemize}
    \item $\frac{\partial F_i}{\partial x^j}=-\frac{1}{r^2}\,e^A_{i,j}-\frac{2\,x_j}{r^2}\,F_i$
    \item $\frac{\partial G^1}{\partial x^j}=\frac{1}{r^2}\,G^2_j-\frac{4\,x_j}{r^2}\,G_1$
    \item $\frac{\partial}{\partial x^j}\left (x_i \, G^1 \right )= \frac{x^i}{r^2}\,G^2_j + \left (\delta_i^j - \frac{4\,x_i\,x_j}{r^2} \right)\,G^1$
    \item $\frac{\partial G^2_i}{\partial x^j}=-\frac{2\,x_j}{r^2}\,G^2_i + \frac{i}{r^2}\,\left (e^S_{i,j}-\frac{2\,\delta_i^j}{n}\,Id \right )$
    \item $[e^A_{k,i},e^A_{l,j}]=\delta_i^l\,e^A_{k,j} + \delta_i^j\,e^A_{l,k} - \delta_k^l\,e^A_{i,j}-\delta_k^j\,e^A_{l,i}$
    \item $[i\,e^S_{k,i},i\,e^S_{l,j}]=-\delta^l_i\,e^A_{k,j}+\delta^j_i\,e^A_{l,k}-\delta_l^k\,e^A_{i,j}+\delta_k^j\,e^A_{l,i}$
    \item $[e^A_{k,i},i\,e^S_{l,j}]=\delta_i^l\,i\,e^S_{k,j} + \delta_i^j\,i\,e^S_{k,l}-\delta_k^l\,i\,e^S_{i,j}-\delta_k^j\,i\,e^S_{i,l}$
    \item $[F_i,F_j]=-\frac{1}{r^2}\,e^A_{i,j}+\frac{1}{r^2}\, \left (x_i\,F_j-x_j\,F_i\right )$
    \item $[G^2_i,G_j^2]=-\frac{1}{r^2}\,e^A_{i,j} -\frac{1}{r^2}\, \left (x_i\,F_j-x_j\,F_i\right )$
    \item $[F_i,G^2_j]=-\frac{1}{r^2}\,i\,e^S_{i,j}+2\,\frac{\delta_i^j}{r^4}\,i\,x\cdot x^T +\frac{1}{r^2}\,\left (x_i\,G^2_j-x_j\,G^2_i \right)$
    \item $[F_i,G^1]=\frac{2\,x_i}{r^2}\,G^1-\frac{1}{r^2}\,G^2_i$
    \item $[G^2_i,G^1]=\frac{1}{r^2}\,F_i$
\end{itemize}
\end{lemma}

With this in mind we can then compute the curvature:

\begin{proposition}[Curvature of an $SO(n)$ equivariant $SU(n)$ connection]
The curvature we obtain is:
$$
F_{0,i}=\frac{\partial h_1}{\partial t}\,F_i + \frac{\partial h_3}{\partial t}\,x_i\,G^1 + \frac{\partial h_2}{\partial t}\,(G^2_i-2\,x_i\,G^1) 
$$
and
$$
\begin{aligned}
F_{i,j}=&\left (r\,\frac{\partial h_1}{\partial r}+h_1^2+h_2^2-1-h_2\,h_3 \right )\,\frac{1}{r^2}\,\left (x_i\,F_j-x_j\,F_i \right) \\
+& \left ( 1-h_1^2-h_2^2 \right )\,\frac{1}{r^2}\,e^A_{i,j} \\
+& \left (r\,\frac{\partial h_2}{\partial r} +h_1\,h_3 \right) \, \frac{
i}{r^2}\,\left (x_i\,G^2_j-x_j\,G^2_i \right )
\end{aligned}
$$
\end{proposition}

Using proposition \ref{1st} we can write Yang-Mills in our case as:
$$
\begin{cases}
\sum_{j=1}^n [B_j,F_{0,j}] + \frac{\partial F_{0,j}}{\partial x^j} + \left ( \frac{(n-2)\,f_r'-f_t'}{2\,r}\right )\,x_j\,F_{0,j} =0 \\
e^{-f_t+f_r}\,\frac{\partial F_{i,0}}{\partial t} - \left (\sum_{j=1}^n [B_j,F_{i,j}] + \frac{\partial F_{i,j}}{\partial x^j} + \left (\frac{f_t'+(n-4)\,f_r'}{2\,r}\right ) \, x_j\,F_{i,j} \right)=0 & \forall i\\
\end{cases}
$$

Therefore, to compute each term we will need the following formulas:

\begin{lemma}[Summation computational help for $SU(n)$]
We have the following formulas:
\begin{itemize}
    \item $\sum_{j=1}^n x_j\,F_j=0$
    \item $\sum_{j=1}^n x_j\,G^2_j=2\,r^2\,G^1$
    \item $\sum_{j=1}^n \frac{\partial F_j}{\partial x^j}=0$
    \item $\sum_{j=1}^n \frac{\partial G^2_j}{\partial x^j}=-4\,G^1$
    \item $\sum_{j=1}^n \frac{\partial}{\partial x^j}(x_j\,G^1)=(n-2)\,G^1$
    \item $\sum_{j=1}^n [F_j,G_j^2]=2\,n\,G^1$
    \item $\sum_{j=1}^n [F_j,x_j\,G^1]=\sum_{j=1}^n [G^2_j,x_j\,G^1]=0$
\end{itemize}
\end{lemma}

We then get that:
\begin{theorem}[$SO(n)$-equivariant Yang-Mills for an $SU(n)$-connection]
If we take an $SO(n)$-equivariant $SU(n)$-connection as described earlier, then Yang-Mills can be rewritten as the following equations:
\begin{equation}
    \begin{aligned}\label{timeq}
    &2\,n\,\left (h_1\,\frac{\partial h_2}{\partial t}-h_2\,\frac{\partial h_1}{\partial t} \right)+ \frac{\partial}{\partial t} \left ( r\,\frac{\partial h_3}{\partial r} + (n-2)\,h_3 + \left (\frac{(n-2)\,f_r'-f_t'}{2} \right )\,r\,h_3\right )=0
    \end{aligned}
\end{equation}
\begin{equation}\label{h_1wave}
    \begin{aligned}
    &\frac{\partial^2 h_1}{\partial r^2}-e^{-f_t+f_r}\,\frac{\partial ^2 h_1}{\partial t^2} + \left (\frac{f_t'+(n-4)\,f'_r}{2} \right ) \,\left ( \frac{\partial h_1}{\partial r}-\frac{h_2\,h_3}{r} \right) +\frac{n-3}{r}\,\frac{\partial h_1}{\partial r}\\
    -& \frac{1}{r}\,\frac{\partial}{\partial r}(h_2\,h_3)  -\frac{n-4}{r^2}\,h_2\,h_3+\frac{n-2}{r^2}\,h_1\,(1-h_1^2-h_2^2)-h_3\, \left ( \frac{1}{r}\,\frac{\partial h_2}{\partial r}+\frac{h_1\,h_3}{r^2} \right ) =0
    \end{aligned}
\end{equation}
\begin{equation}\label{h_2wave}
    \begin{aligned}
    &\frac{\partial^2 h_2}{\partial r^2}-e^{-f_t+f_r}\,\frac{\partial^2 h_2}{\partial t^2} + \left( \frac{f_t'+(n-4)\,f_r'}{2} \right ) \, \left ( \frac{\partial h_2}{\partial r}+\frac{h_1\,h_3}{r^2} \right ) +\frac{n-3}{r}\,\frac{\partial h_2}{\partial r}\\
    +& \frac{1}{r}\,\frac{\partial}{\partial r}(h_1\,h_3) + \frac{n-4}{r^2}\,h_1\,h_3 +\frac{n-2}{r^2}\, h_2\,\left (1-h_1^2-h_2^2 \right )+ h_3\,\left ( \frac{1}{r}\,\frac{\partial h_1}{\partial r}-\frac{h_2\,h_3}{r^2} \right ) =0
    \end{aligned}
\end{equation}
\begin{equation}\label{h_3wave}
    \begin{aligned}
    e^{-f_t+f_r}\,\frac{\partial^2 h_3}{\partial t^2}-2\,n\,\left ( \frac{1}{r}\,\left ( h_2\,\frac{\partial h_1}{\partial r}-h_1\,\frac{\partial h_2}{\partial r} \right ) - \frac{1}{r^2}\,h_3\left ( h_2^2+h_1^2 \right ) \right )=0
    \end{aligned}
\end{equation}

We will write $\mathrm{S}$ the space of solutions.
\end{theorem}
\begin{remark}[Observations]
Here are a few observations we can make from the get-go:
\begin{itemize}
    \item Equations \ref{h_1wave} and \ref{h_2wave} are two wave equations that are almost completely the same for $h_1$ and $h_2$, except for the fact that one is 'perturbed' by $h_3$ and the other by $-h_3$.
    \item Equation \ref{h_3wave} can be seen as a periodic LODE on $h_3$ on time.
    \item The system is equivalent under the transformations $(h_1,h_2,h_3)\mapsto (h_2,h_1,-h_3)$ and with a constant $\theta$, $(h_1,h_2,h_3)\mapsto(\cos(\theta)\,h_1-\sin(\theta)\,h_2,\sin(\theta)\,h_1+\cos(\theta)\,h_2,h_3)$. Therefore $O(2)$ acts on the space of solutions with:
    $$
\begin{aligned}
    O(2)\times \mathrm{S} &\longrightarrow \mathrm{S} \\
    \left(\Lambda, \left ( \begin{pmatrix}
     h_1 \\
     h_2
    \end{pmatrix} , h_3 \right ) \right ) &\longmapsto \left( \Lambda\cdot \begin{pmatrix}
     h_1 \\
     h_2
    \end{pmatrix} , det(\Lambda)\,h_3 \right )
\end{aligned}
    $$
\end{itemize}
\end{remark}
Using these symmetries, let us write $h=h_1+i\,h_2=|h|e^{i\,\phi}$ we then get that the system is equivalent to the following wave equation:
\begin{equation}\label{hwave}
    \begin{aligned}
    &\frac{\partial^2 h}{\partial r^2}-e^{-f_t+f_r}\,\frac{\partial ^2 h}{\partial t^2}+\frac{n-2}{r^2}\,h\,(1-|h|^2)\\
    +& \left (\frac{f_t'+(n-4)\,f'_r}{2} +\frac{n-3}{r}+i\,\frac{h_3}{r}\right ) \,\left ( \frac{\partial h}{\partial r}+i\,\frac{h\,h_3}{r} \right)+ i\,\frac{\partial}{\partial r}(\frac{h_3}{r}\,h) =0
    \end{aligned}
\end{equation}
Where $h_3$ is under the constraints:
\begin{equation}\label{complextime}
    \begin{aligned}
    \frac{\partial^2}{\partial t\,\partial r}\left (r^{n-2}\,exp\left ( \frac{(n-2)\,f_r-f_t}{2} \right ) \,h_3 \right ) = 2\,n\,r^{n-3}\,exp\left ( \frac{(n-2)\,f_r-f_t}{2} \right )\,|h|^2\,\frac{\partial \phi}{\partial t}
    \end{aligned}
\end{equation}
and
\begin{equation}\label{complexh_3wave}
  \begin{aligned}
  e^{-f_t+f_r}\,\frac{\partial^2 h_3}{\partial t^2} + 2\,n\,\frac{|h|^2}{r^2}\,h_3 = -2\,n\,\frac{|h|^2}{r}\,\frac{\partial \phi}{\partial r}
  \end{aligned}  
\end{equation}

We can then rewrite the equations and make apparent the relations to a wave equation on $|h|$:
\begin{proposition}[Wave equation rewriting]
The first equation is equivalent to:
\begin{equation}
    \begin{aligned}\label{|h|wavemap}
    &\left (\frac{\partial^2 |h|}{\partial r^2}-e^{-f_t+f_r}\,\frac{\partial ^2 |h|}{\partial t^2}
    +\frac{\partial}{\partial r} \left (\frac{f_t+(n-4)\,f_r}{2} +\log(r^{n-3})\right ) \,\frac{\partial |h|}{\partial r}+\frac{n-2}{r^2}\,|h|\,(1-|h|^2) \right )\,e^{i\,\phi} \\
    =&i\, ( e^{-f_t+f_r} \left [2\,\frac{\partial |h|}{\partial t}\,e^{i\,\phi}+h\,\frac{\partial}{\partial t} \right ]\,\frac{\partial \phi}{\partial t}\\
    &- \left [\frac{\partial}{\partial r} \left (\frac{f_t+(n-4)\,f_r}{2} +\log(r^{n-3})\right ) + i\,\left ( \frac{\partial \phi}{\partial r} + \frac{h_3}{r}\right ) +\frac{\partial}{\partial r} \right ] \, h\,\left (\frac{\partial \phi}{\partial r} + \frac{h_3}{r}\right )  )
    \end{aligned}
\end{equation}
\end{proposition}

\begin{proposition}[Static solutions]
If we suppose that our system is static and that $h\neq 0$ then $\phi$ is free and we have:
$$
\begin{cases}
h_3=-r\,\frac{\partial \phi}{\partial r} \\
\frac{\partial^2|h|}{\partial r^2}+\frac{\partial}{\partial r}\left (\frac{f_t+(n-4)\,f_r}{2}+\log(r^{n-3}) \right )\,\frac{\partial |h|}{\partial r} + \frac{n-2}{r^2}\,|h|\,(1-|h|^2) =0
\end{cases}
$$
\end{proposition}

\begin{proposition}[Constant solutions]
If we suppose that our system is constant then we have either $h=0$ and $h_3$ is free or $h_3=0$ and $|h|=1$.
\end{proposition}

\begin{proposition}[Time invariance of $h_3$]
If we suppose:
$$
\begin{cases}
\frac{\partial \phi}{\partial t}=0 \\
\frac{\partial \phi}{\partial r} + \frac{h_3}{r}=0
\end{cases}
$$
Or
$$
\begin{cases}
\frac{\partial \phi}{\partial t}=0 \\
\frac{\partial^2h_3}{\partial t^2}=_{r\rightarrow 0} o\left(r^{n-2}\,exp \left (\frac{(n-2)\,f_r-f_t}{2}\right)\right)
\end{cases}
$$
Or equivalently, simply:
$$
\frac{\partial h_3}{\partial t}=0
$$
Then the system reduces to $r\mapsto \phi(r)$ being free and 
\begin{equation}\label{wavemap}
\frac{\partial^2 |h|}{\partial r^2}-e^{-f_t+f_r}\,\frac{\partial ^2 |h|}{\partial t^2}
    +\frac{\partial}{\partial r} \left (\frac{f_t+(n-4)\,f_r}{2} +\log(r^{n-3})\right ) \,\frac{\partial |h|}{\partial r}+\frac{n-2}{r^2}\,|h|\,(1-|h|^2)=0
\end{equation}
\end{proposition}

\begin{remark}[Energy of the wave equation]
Notice that equation \ref{wavemap} admits the following positive and conserved energy (in the mathematical sense):
$$
E(t)=\int_0^{\infty} \left [ \left ( \frac{\partial |h|}{\partial r} \right)^2 + e^{-f_t+f_r}\,\left ( \frac{\partial |h|}{\partial t} \right)^2 + \frac{n-2}{2\,r^2}\,\left (|h|^2-1\right )^2 \right]\,r^{n-3}\,exp\left ( \frac{f_t+(n-4)\,f_r}{2}\right )\,dr
$$

And if we consider 
$$
S(r,t)=\begin{pmatrix}
|h| \\
f_t \\
f_r \\
\end{pmatrix}
$$
Similarly to \cite{Roland,Roland2} the transformation $S \mapsto S^{\lambda}(r,t)=S(r/\lambda,t/\lambda)$ gives us:
$$
E^{\lambda}(t)=\lambda^{n-4}\,E(t/\lambda)
$$
Meaning that we have the same criticality as in \cite{Roland} and \cite{Roland2}.
\end{remark}
\newpage

\section{Ouverture}
As we have shown how $G$-equivariance could simplify Yang-Mills, our hope is that the concept of $G$-equivariance can be used more generally and the results used in several domains which we could not by lack of time or expertise:
\begin{enumerate}
    \item Regarding geometric aspects, one could try to show conjecture \ref{conj}. One could also use the principle of symmetric criticality \cite{Pal} as it was used in \cite{Gastel2011YangMillsCO} to prove the topological and geometrical coherence and existence of the objects we worked with.
    \item Regarding physics, one could look at the equations obtained in section \ref{Iso} and couple them with Einstein's equations to determine $f_t$ and $f_r$ in the right context. One could also test the validity of the models regarding the weak interaction for a single particle, the strong interaction for a globally neutral hadron and the electroweak force for a particle at a fixed point of time (see section \ref{electroweak}).
    \item Regarding the equations we obtain, one could look at the stability of some of their special solutions to gain further information and intuition on the stability of Yang-Mills' Cauchy problem with singularities.
    \item Finally, regarding the method used, one could apply it to several other systems whenever there exists a group action on the base and on the final space.
\end{enumerate}
Et "je ne sais pas le reste."\footnote{Evariste Galois}
\section{Acknowledgements}
Finally, this work would not have been possible if Roland DONNINGER had not seen that there was a missing bridge in justifying an equivariant ansatz. Furthermore, the internship leading to this paper was made possible by Anne-Laure DALIBARD who was able to ascertain the common interest of two people a country away from each other.
\newpage

\bibliography{ibliography}

\newpage
\appendix
\section{Notations and glossary}
\begin{table}[h!]
\begin{center}
\begin{tabular}{c c}
     $\mathfrak{X}(M)$ & The vector space of vector fields over $M$   \\
     $\mathfrak{g}$ & The lie algebra associated to the lie group $G$\\
     $Fl^{Y}_s$ & The flow of $Y$ evaluated at time $s$ \\
     $R_g$ & $R_g (h)=h\cdot g$ \\
     $L_g$ & $L_g (h)=g\cdot h$ \\
     $Ad(g)$ & $d_e (conj_g)$ \\
     $\Omega^k (M,A)$ & The space of $k$-forms with values in $A$\\
     $f^* X$ & The pullback of $X$ with respect to $f$\\
     $ker(\phi)$ & $\{x|\phi (x) =0\}$ \\
     $\epsilon_{i,j,..,k}$ & $\begin{cases}
     sign(\sigma) & (i,j,...,k)=\sigma(1,2,3,...) \\
     0
     \end{cases}$ \\
     $Tconj(P)$ & $Tconj(P)(M)=P\cdot M \cdot P^T$ \\
     $\delta_i^j$ & $\begin{cases}
     1 & i=j \\
     0
     \end{cases}$\\
     $g_{i,j}^{\theta}$ & $\begin{pmatrix}
I_{i-1} & & \\
 & \begin{matrix}
  \cos{\theta} & & -\sin{\theta} \\
   & I_{j-i-1} & \\
   \sin{\theta} & & \cos{\theta} 
 \end{matrix} & \\
 & & I_{n-j}
\end{pmatrix}$ \\
$k_{i,j}^h$ & $\begin{pmatrix}
I_{i-1} & & \\
 & \begin{matrix}
  \cosh{h} & & \sinh{h} \\
   & I_{j-i-1} & \\
   \sinh{h} & & \cosh{h} 
 \end{matrix} & \\
 & & I_{n-j}
\end{pmatrix}
$ \\
$M_{p,q}$ &  $\{v=(t,x)\in \mathbf{R}^{p+q} | v^T\cdot I_{p,q} \cdot v>0 \}$
\end{tabular}
\end{center}
\end{table}
\end{document}